\documentclass[journal,twoside,web]{ieeecolor}

\usepackage{generic}
\usepackage{textcomp}

\usepackage{amsmath} 
\usepackage{amssymb}  
\usepackage{graphicx}
\usepackage{epstopdf}
\usepackage{amsmath}
\usepackage{mathtools}
\usepackage{dsfont}
\usepackage{tikz}
\usepackage{siunitx}
\usepackage{hyperref}

\DeclareMathOperator*{\argmin}{argmin}
\DeclareMathOperator*{\argmax}{argmax}

\usepackage[bottom]{footmisc}

\usepackage{balance}    
\usepackage{amsthm}
\usepackage[ruled,vlined,linesnumbered]{algorithm2e}
\usepackage{xcolor}
\usepackage{tcolorbox}
\usepackage{algpseudocode}
\usepackage{dsfont}

\usepackage[noadjust]{cite} 
\usepackage{amsmath,stmaryrd,graphicx}
\usepackage{mathtools}

\DeclarePairedDelimiter\floor{\lfloor}{\rfloor}

\newtheoremstyle{boldStyle}
  {\topsep}
  {\topsep}
  {\itshape}
  {0pt}
  {\bfseries}
  {.}
  { }
  {\thmname{#1}\thmnumber{ #2}\thmnote{ (#3)}}

\newtheoremstyle{italicStyle}
  {\topsep}
  {\topsep}
  {}
  {0pt}
  {\bfseries}
  {.}
  { }
  {\thmname{#1}\thmnumber{ #2}\thmnote{ (#3)}}

\theoremstyle{boldStyle}
\newtheorem{theorem}{Theorem}

\newtheorem{proposition}{Proposition}

\theoremstyle{italicStyle}
\newtheorem{assumption}{Assumption}
\newtheorem{remark}{Remark}

\newcommand{\Gcurr}{\mathcal{C}_\textrm{curr}^k}
\newcommand{\Ggoal}{\mathcal{C}_\textrm{goal}^k}
\newcommand{\Gforc}{\mathcal{C}_\textrm{forc}^k}
\newcommand{\Xgoal}{\mathcal{X}^k_\textrm{goal}}
\newcommand{\Xcurr}{\mathcal{X}^k_\textrm{curr}}
\newcommand{\pgoal}{p_\textrm{goal}^k}

\newcommand{\GcurrMinus}{\mathcal{C}_\textrm{curr}^{k-1}}
\newcommand{\GgoalMinus}{\mathcal{C}_\textrm{goal}^{k-1}}

\newcommand{\pgoalMinus}{p_\textrm{goal}^{k-1}}
\newcommand{\XgoalMinus}{\mathcal{X}^{k-1}_\textrm{goal}}

\newcommand{\GcurrPlus}{\mathcal{C}_\textrm{curr}^{k+1}}
\newcommand{\GgoalPlus}{\mathcal{C}_\textrm{goal}^{k+1}}

\newcommand{\XgoalPlus}{\mathcal{X}^{k+1}_\textrm{goal}}
\newcommand{\XcurrPlus}{\mathcal{X}^{k+1}_\textrm{curr}}
\newcommand{\pgoalPlus}{p_\textrm{goal}^{k+1}}

\newcommand{\Zp}{\mathbb{Z}_{0+}}
\newcommand{\Rp}{\mathbb{R}_{0+}}
\newcommand{\be}{b}
\newcommand{\Be}{\mathcal{B}}

\def\sq{\mathbin{{\strut\rule{1.25ex}{1.25ex}}}}
\renewenvironment{proof}{{\textbf{Proof:}}}{\hfill$\sq$}

\definecolor{ugoColor}{rgb}{0.6,0.8,0.0}
\definecolor{ligthGray}{rgb}{0.95,0.95,0.95}

\makeatletter
\newcommand{\fixed@sra}{$\vrule height 2\fontdimen22\textfont2 width 0pt\shortrightarrow$}
\newcommand{\shortarrow}[1]{%
  \mathrel{\text{\rotatebox[origin=c]{\numexpr#1*45}{\fixed@sra}}}
}
\makeatother

\begin{document}


\title{Unified Multi-Rate Control: from Low-Level Actuation to High-Level Planning}

\author{
Ugo Rosolia,~\IEEEmembership{Member,~IEEE,} Andrew Singletary,~\IEEEmembership{Member,~IEEE,} and~Aaron D. Ames,~\IEEEmembership{Fellow,~IEEE}
\thanks{U. Rosolia, A. Singletary and A. D. Ames are with the AMBER lab at the California Institute of Technology, Pasadena,
CA, USA, e-mail: \texttt{\{urosolia, asinglet, ames\}@caltech.edu}. The authors would like to acknowledge the support by the National Science Foundation award \#1932091.}
}


\maketitle

\begin{abstract}
In this paper we present a hierarchical multi-rate control architecture for nonlinear autonomous systems operating in partially observable environments. 
Control objectives are expressed using syntactically co-safe Linear Temporal Logic (LTL) specifications and the nonlinear system is subject to state and input constraints. 
At the highest level of abstraction, we model the system-environment interaction using a discrete Mixed Observable Markov Decision Process (MOMDP), where the environment states are partially observed. The high-level control policy is used to update the constraint sets and cost function of a Model Predictive Controller (MPC) which plans a reference trajectory. Afterwards, the MPC planned trajectory is fed to a low-level high-frequency tracking controller, which leverages Control Barrier Functions (CBFs) to guarantee bounded tracking errors. 
Our strategy is based on model abstractions of increasing complexity and layers running at different frequencies.
We show that the proposed hierarchical multi-rate control architecture maximizes the probability of satisfying the high-level specifications while guaranteeing state and input constraint satisfaction. 
Finally, we tested the proposed strategy in simulations and experiments on examples inspired by the Mars exploration mission, where only partial environment observations are available.
\end{abstract}


\IEEEpeerreviewmaketitle

\section{Introduction}

Control design for complex cyber-physical systems, which are described by continuous and discrete variables, is usually divided into different layers~\cite{wongpiromsarn2011tulip, Wongpiromsarn2012, wongpiromsarn2010receding, tabuada2006linear, alur2000discrete, haesaert2019temporal, haesaert2018temporal, kousik2018bridging, shao2020reachability, herbert2017fastrack, csomay2022multi}. 
Each layer is designed using model of increasing accuracy and complexity, which allow the controller to take high-level decisions--e.g., perform an overtaking maneuver--and to compute low-level commands--e.g., the input current to a motor.
High-level decisions and low-level control actions are computed at different frequencies and the interaction between layers should be taken into account to guarantee safety of the closed-loop system~\cite{Wongpiromsarn2012}.

In this work, we present a multi-rate hierarchical control scheme for  nonlinear systems operating in partially observable environments. Our architecture, which is composed by three layers running at different frequencies, guarantees constraint satisfaction and maximization of the closed-loop probability of satisfying the high-level specifications. At the lowest level, we leverage the continuous time nonlinear system model to guarantee a bounded tracking error. The mid-level planning layer computes a reference trajectory using a simplified prediction model and the low-level tracking error bounds. Finally, at the highest level of abstraction we model the system-environment interaction using Mixed Observable Markov Decision Processes (MOMDPs), which allows us to account for partial environment observations.

\begin{figure}[t!]
    \centering
	\includegraphics[trim=0mm 0mm 0mm -4mm, width= 1.0\columnwidth]{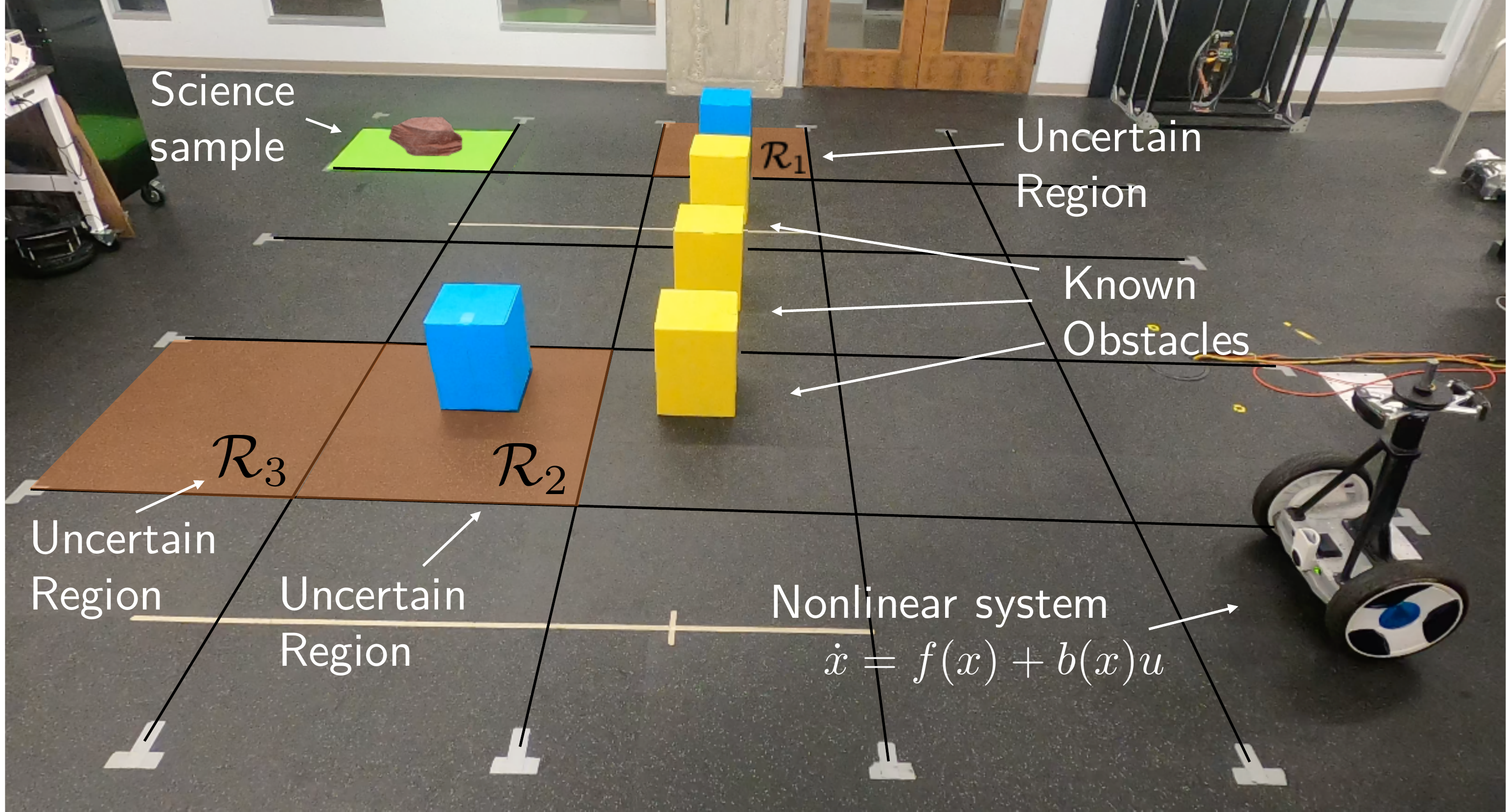}
    \centering
    \caption{This figure shows an environment composed of $25$ cells, $3$ obstacles (yellow and blue boxes) and $3$ uncertain regions (light brown). In this example the goal of the controller is to explore the state space in order to find a science sample.}
    \label{fig:envScheme}
\end{figure}

\begin{figure*}[h!]
    \includegraphics[trim= 0mm 0mm 0mm 0mm, clip, width=0.97\textwidth]{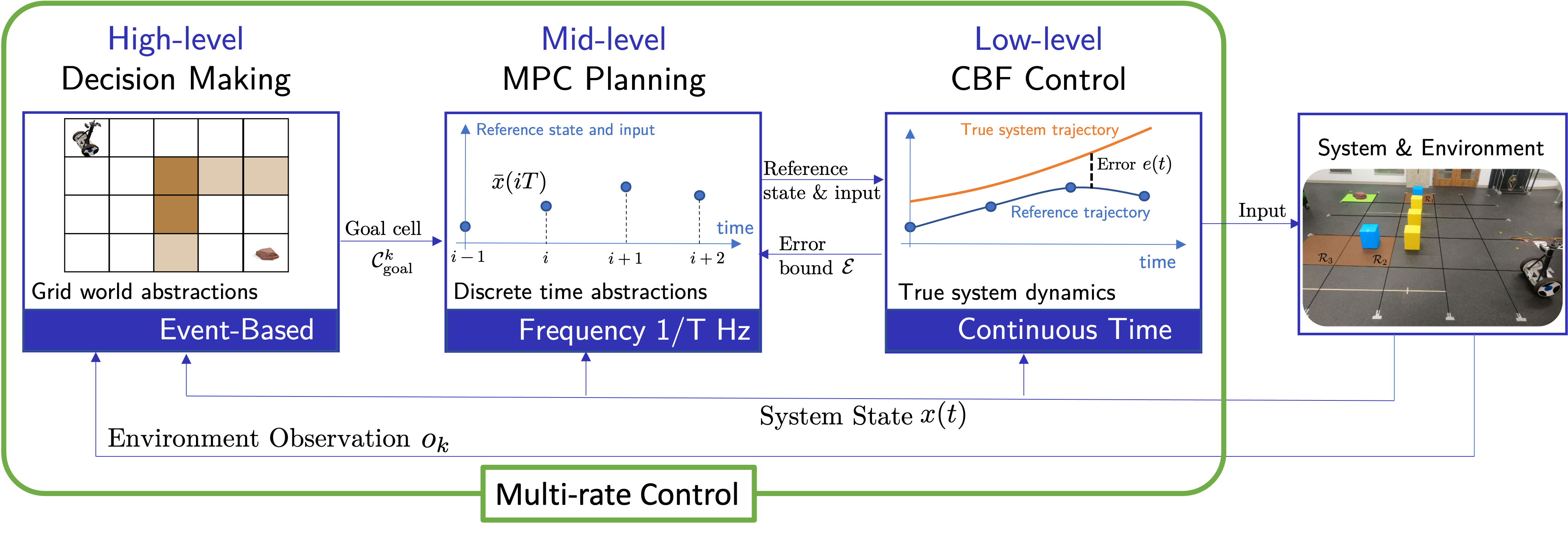}
  \caption{Multi-rate control architecture. The high-level decision maker leverages the system's state $x(t)$ and partial environment observations $o_k$ to compute a goal cell, the constraint set and the goal positions, which are fed to the mid-level MPC planner. The planner computes a reference trajectory given the tracking error bounds $\mathcal{E}$ from the low-level tracking controller. Finally at the lowest level, the control action is computed summing up the mid-level input $u_m(t)$ and the low-level input $u_l(t)$.}\label{fig:summary}
\end{figure*}

\subsection{Related Work}
Control policies for high-level decision making are usually synthesized using discrete model abstractions. The high-level control objectives are often expressed by Linear Temporal Logic (LTL) formulas~\cite{pnueli1977temporal}, as they are a formalism to express high-level system behaviors using logical and temporal operators~\cite{pnueli1977temporal}. Motion planning with LTL and syntactically co-safe LTL (scLTL) specifications has been widely studied in literature~\cite{loizou2004automatic, wongpiromsarn2010receding, fainekos2005hybrid, kloetzer2008fully, Wongpiromsarn2012, tabuada2006linear, haesaert2018temporal ,haesaert2019temporal, nilsson2018toward,bouton2020point,vasile2016control,wang2018bounded,ahmadi2020stochastic,kwiatkowska2011prism, dehnert2017storm}. For deterministic systems with finite-state spaces several approaches and toolboxes are available for synthesis~\cite{loizou2004automatic, fainekos2005hybrid, kloetzer2008fully, Wongpiromsarn2012, tabuada2006linear,wongpiromsarn2010receding}. 
When the system-environment interaction are uncertain, the high-level abstractions are described by discrete Markov Decision Processes (MDPs) and the high-level decision making problem can be solved exactly using dynamic programming, policy iteration, and linear programming strategies~\cite{altman1999constrained}. On the other hand, when the system dynamics are uncertain and only partial observations are available, the system-environment interaction can be modeled using discrete Partially Observable Markov Decision Processes (POMPDs). Computing a control policy in POMDPs settings is NP-hard~\cite{sondik1978optimal}, but approximate solutions can be computed using finite state controllers~\cite{PoupartB03} and performing point-based approximations~\cite{pineau2003point}.

Given a high-level decision, reachability-based techniques~\cite{wongpiromsarn2010receding,Wongpiromsarn2012} or simulation-based abstractions~\cite{tabuada2006linear,alur2000discrete} may be used to compute a goal set for the continuous time system, e.g., a subset of a lane where we would like to drive the vehicle when performing an overtaking maneuver. 
Therefore, the input to the system's actuators is computed solving mid-level planning and low-level control problems, which have been studied extensively in literature~\cite{gurriet2018towards, CBF, wang2017safety, wabersich2018linear, herbert2017fastrack,yin2019optimization, singh2018robust, singh2017robust, gao2014tube, kogel2015discrete, yu2013tube, gundana2021event,
chen2020automatic}. 
The planning problem is usually defined for a simplified model and the resulting reference trajectory is then tracked using low-level controllers, which leverage the nonlinear system dynamics. Tracking controllers may be synthesized using Hamilton-Jacobi (HJ) reachability analysis~\cite{herbert2017fastrack} or sum-of-squares programming~\cite{singh2018robust,yin2019optimization}.
Another strategy to solve mid-level planning and low-level control problems is to use nonlinear tube MPC~\cite{gao2014tube, kogel2015discrete, yu2013tube, singh2017robust,kohler2020computationally}, where the difference between the planned trajectory and the actual one is over approximated using Lyapunov based analysis or Lipschitz properties of the nonlinear dynamics.
When the planned trajectory is computed without taking into account tracking errors, safety can be guaranteed using filters which, given a desired mid-level command, compute a safe control action using CBFs~\cite{gurriet2018towards,CBF,wang2017safety}, feasibility of an MPC problem~\cite{wabersich2018linear}, or reachability analysis~\cite{shao2020reachability}. A different strategy that can be used to bridge the gap between high-level decision making and low-level control is to leverage CBFs~\cite{gundana2021event, chen2020automatic}. However, these strategies compute control actions without forecasting the evolution of the system's trajectory and they may result in sub-optimal behaviors. 

As discussed next, our approach leverages both MPC and CBF policies to compute low-level commands given high-level decisions. The forcast from the MPC planning layers is used to compute a feedforward term that allows us to mitigate the myopic nature of CFBs, which are used to guarantee safety at the continuous time layer~\cite{CBF}.

\subsection{Contribution}
Our contribution is threefold. First, we introduce a mid-level planner that leverages two MPC problems with time-varying constraint sets and cost functions. These time-varying components are given by the high-level decision maker and they can jeopardize the feasibility of standard MPC schemes, which are designed to steer the system to a time-invariant goal state. For instance, the safety guarantees from~\cite{multiRosolia}--\cite{mayne2005robust} are lost when the goal state and constraints are updated online during the execution of the control task. To overcome limitations of standard time-invariant approaches, we propose a contingency scheme where at each time step we solve at most two MPC problems. This strategy guarantees feasibility of the planner with time-varying components. In particular, in Algorithm~\ref{algo:multiRate} we introduce a contingency MPC problem that is defined by updating the time-varying components as a function of the latest planned optimal trajectory. 

Second, we show how to integrate a CLF-CBF tracking controller with a mid-level planner. We present a constraint tightening approach that accounts for the low-level tracking error and we demonstrate that the resulting multi-rate control architecture guarantee safety, when a local reachability assumption on the system dynamics is satisfied. Such reachability assumption, which is tailored to navigation problems, together with the proposed contingency scheme allows us to avoid the construction of finite state abstractions defined over the entire state space.
Furthermore compared to the constraint tightening from our previous work~\cite{multiRosolia}, the proposed constraint tightening builds upon ideas from the fixed-tube robust MPC strategy~\cite{mayne2005robust}, where the initial state of the planned trajectory is an optimization variable. For this reason, the proposed constraint tightening does not require the online computation of robust reachable sets to formulate the MPC problem and therefore it is computationally more efficient than the formulation from~\cite{multiRosolia}.

Third, we show how to model the system-environment interaction using Mixed Observable Markov Decision Processes (MOMDPs), where the system's state is fully observable and the environment's state is partially observable. We build upon~\cite{rosolia2021time} that introduced the synthesis process for systems with discrete state and action spaces, and we focus on the multi-layer hierarchical control design for systems with continuous states and actions. In particular, we show how to leverage high-level decisions from the MOMDP to construct the MPC time-varying components, and we demonstrate that the proposed multi-layer hierarchical control strategy maximizes the probability of satisfying the high-level specifications. Finally, we test our strategy on navigation tasks as the one shown in Figure~\ref{fig:envScheme}, where a Segway like-robot has to find science samples while navigating a partially observable environment. 


This paper is organized as follows. The background material is  discussed in Section~\ref{sec:prel}. Section~\ref{sec:probForm} describes the problem under study. First, we introduce the system and environment models, and afterwards the control design objectives. The hierarchical architecture is introduced in Section~\ref{sec:proposedArchitecture}, where we present the high-level decision maker, the mid-level MPC planer, and the low-level CBF controller. The closed-loop properties are discussed in Section~\ref{sec:properties}. Finally, we illustrate the effectiveness of the proposed strategy with high-fidelity simulations and hardware experiments.
 
\section{Preliminaries}\label{sec:prel}%

\noindent
\textbf{Notation: }
The Minkowski sum of two sets $\mathcal{X}\subset \mathbb{R}^{n_x}$ and $\mathcal{Y}\subset \mathbb{R}^{n_x}$ is denoted as $\mathcal{X}\oplus\mathcal{Y}$, and the  Pontryagin difference as $\mathcal{X}\ominus\mathcal{Y}$. $\mathcal{K}^e$ is the set of extended class-$\mathcal{K}^e$ functions $\beta$ which are strictly increasing and $\beta(0)=0$. For a set $\mathcal{A}\subset\mathbb{R}^{n_x}$ and a vector $x\in\mathbb{R}^{n_x}$, we denote the projection
\begin{equation*}
    \text{Proj}(x,\mathcal{A}) = \argmin_{d \in \mathcal{A}} ||x - d ||_2,
\end{equation*}
and the cardinality of the set $\mathcal{A}$ as $|\mathcal{A}|$. We define $\Zp = \{0, 1,2,\ldots\}$ and $\Rp = \{x\in\mathbb{R}^{n_x}|x\geq0\}$ which denote the set of positive integers and real numbers, respectively. Finally, given $t \in \Rp$ and $T \in \Zp$ we define $\floor{t/T}=\text{floor}(t/T)$.

\vspace{2pt}
\noindent\textbf{Specifications: } High-level objectives are expressed using syntactically co-safe Linear Temporal Logic (scLTL) specifications. For a set of atomic proposition $\mathcal{AP}$, an scLTL specification is defined as follows:
\begin{equation*}
    \psi : =  p ~ | ~ \neg p ~ | ~ \psi_1 \land \psi_2 ~ | ~ \psi_1 \lor \psi_2 ~ | ~ \psi_1 U \psi_2~ | ~ \bigcirc \psi ,
\end{equation*}
where the atomic proposition $p \in \mathcal{AP}$ and $\psi, \psi_1, \psi_2$ are scLTL formulas, which can be defined using the logic operators negation ($\neg$), conjunction ($\land$) and disjunction ($\lor$). Furthermore, scLTL formulas can be specified using the temporal operators until ($U$) and next ($\bigcirc$). Each atomic proposition $p$ is associated with a subset of the high-level state space $\mathcal{P}$ and a high-level state $\omega_k$ satisfies the proposition $p$ if $\omega_k \in \mathcal{P}$. Finally, satisfaction of a specification $\psi$ for the trajectory $\boldsymbol{\omega}_k = [\omega_k, \omega_{k+1}, \ldots]$, denoted by  
\begin{equation}\label{eq:specFormula}
    \boldsymbol{\omega}_k \models \psi,
\end{equation}
is recursively defined as follows: $i)$ $\boldsymbol{\omega}_k \models p \iff$ $\omega_k \in \mathcal{P}$, $ii)$ $\boldsymbol{\omega}_k \models \psi_1 \land \psi_2 \iff$ $(\boldsymbol{\omega}_k \models \psi_1)\land (\boldsymbol{\omega}_k \models \psi_1)$, $iii)$ $\boldsymbol{\omega}_k \models \psi_1 \lor \psi_2 \iff$ $(\boldsymbol{\omega}_k \models \psi_1)\lor (\boldsymbol{\omega}_k \models \psi_1)$, $iv)$ $\boldsymbol{\omega}_k \models \psi_1 U \psi_2 \iff$    $\boldsymbol{\omega}_l \models \psi_2$ and $\boldsymbol{\omega}_j \models \psi_2,~\forall j \in \{k, \ldots, l-1\}$, $v)$ $\boldsymbol{\omega}_k \models \bigcirc \psi \iff$ $\boldsymbol{\omega}_{k+1} \models \psi$. Please refer to~\cite[Chapter~3]{belta2017formal} for further details.

\section{Problem Formulation}\label{sec:probForm}
This section describes the problem formulation. First, we introduce the continuous system dynamics. Afterwards, we present the discrete environment model. Finally, we describe the synthesis goals and we summarize the overall control architecture from Figure~\ref{fig:summary}.

\vspace{2pt}
\noindent \textbf{System Model}:
As discussed in the introduction, our goal is to design a controller for nonlinear dynamical systems. In particular, we consider nonlinear control affine systems of the following form:
\begin{equation}\label{eq:sysModel}
    \dot x(t) = f\big( x(t) \big) + g\big( x(t) \big) u(t) ,
\end{equation}
where $f$ and $g$ are Lipschitz continuous, the input $u(t) \in \mathbb{R}^{n_u}$ and the state $x(t) = [p^\top(t), q^\top(t)]^\top
\in \mathbb{R}^{n_x}$
for the position vector $p(t) \in \mathbb{R}^{n_p}$ and the vector $q(t) \in \mathbb{R}^{n_q}$ collecting the remaining states.
Furthermore, the above system is subject to the following state and input constraints:
\begin{equation}\label{eq:lowLevelCnstr}
    u(t) \in \mathcal{U}, p(t_i) \in \mathcal{X}_p \text{ and } q(t_i) \in \mathcal{X}_q,
\end{equation}
for all $ t \in \Rp$ and $t_i = iT$ for all $i \in \Zp$. The time constant $T$ is specified by the user and, as it will be clear later on, it defines the frequency at which the controller updates the planned trajectory. In the above equation~\eqref{eq:lowLevelCnstr}, $\mathcal{X}_p$ represents the free space and $\mathcal{X}_q$ is a user-defined constraint set.

\begin{remark}
    We consider state constraints which are enforced pointwise in time to streamline the presentation. The proposed  control strategy can be extended to account for constraints which must hold for all time $t\in \Rp$. In this case, it is required to modify the low-level controller as discussed in~\cite{multiRosolia}.
\end{remark}

\vspace{2pt}
\noindent\textbf{Environment Model}:
We consider nonlinear dynamical systems operating in partially observable environments, which are partitioned into $\mathcal{C}_1, \ldots, \mathcal{C}_{c}$ cells as in the example from Figure~\ref{fig:envScheme}. We assume that the state of the system is perfectly observable, but we are given only partial observations about the environment state. Thus, at the highest level of abstraction, we model the interaction between the nonlinear system~\eqref{eq:sysModel} and the environment using a Mixed Observable Markov Decision Process (MOMDP). A MOMDP provides a sequential decision-making formalism for high-level planning under mixed full and partial observations~\cite{ong2010planning} and it is defined as a tuple $\left( \mathcal{S}, \mathcal{Z}, \mathcal{A},\mathcal{O}, T_s, T_z, O \right)$, where 
\begin{itemize}

	\item $\mathcal{S}=\{1,\ldots,|\mathcal{S}|\}$ is a set of fully observable states;
	
	\item $\mathcal{Z}=\{1,\ldots,|\mathcal{Z}|\}$ is a set of partially observable states;
    
    \item $\mathcal{A}=\{1,\ldots,|\mathcal{A}|\}$ is a set of actions;

	\item $\mathcal{O}=\{1, \ldots,|\mathcal{O}|\}$ is the set of observations for the partially observable state $z\in \mathcal{Z}$;

	\item The indicator function\footnote{We introduced the indicator function as it will be used later on to compute the belief vector update.} $T_s:\mathcal{S}\times \mathcal{Z} \times \mathcal{A}\times \mathcal{S}\rightarrow \{0, 1\}$ equals one if the system will transition to a state $s'$ given the action $a$ and current state $(s,z)$, i.e.,
    \begin{equation*}
        T_s(s, z, a, s')= \begin{cases}
        1 & \mbox{If } s'=f_s(s,z,a)\\
        0 & \mbox{Else}
        \end{cases},
    \end{equation*}
    where the high-level update function
    $f_s:\mathcal{S}\times\mathcal{Z}\times\mathcal{A}\rightarrow \mathcal{S}$.

    \item The function $T_z: \mathcal{S} \times\mathcal{Z}\times \mathcal{A}\times \mathcal{S}\times \mathcal{Z}\rightarrow [0,1]$ describes the probability of transitioning to a state $z'$ given the action $a$, the successor observable state $s'$, and the system's current state $(s,z)$, i.e.,
    \begin{equation*}
	\begin{aligned}
	    T_z(s&, z, a, s', z')\\
	    &:=\!P(z_{k+1}\!=\!z'|s_{k}\!=\!s, z_{k}\!=\!z,a_{k}\!=\!a, s_{k+1}\!=\!s');
	\end{aligned}
	\end{equation*}
	\item The function $O:\mathcal{S}\times \mathcal{Z} \times \mathcal{A} \times \mathcal{O} \rightarrow [0,1]$ describes the probability of observing the measurement $o \in \mathcal{O}$, given the current state of the system $(s',z') \in \mathcal{S} \times \mathcal{Z}$ and the action $a$ applied at the previous time step, i.e.,
    \begin{equation*}
	    O(s',z', a, o) :=  P(o_k=o|s_{k}=s', z_k= z',a_{k-1}=a);
	\end{equation*}
\end{itemize}

MOMDPs were introduced in~\cite{ong2010planning} to model systems where a subspace of the state space is perfectly observable\footnote{We introduced a special case of the MOMDP from~\cite{ong2010planning} where the transition function of the observable state is deterministic given both the current observable and unobservable states.}. 
In this work, the high-level observable state $s$ represents the location of the system, i.e., the grid cell containing the position vector $p(t)$ which is part of state $x(t) = [p^\top(t), q^\top(t)]^\top$ of the nonlinear system~\eqref{eq:sysModel}. On the other hand, the definition of the partially observable state $z$ depends on the application, and it describes how the environment may affect the evolution of the system. For example, it may be used to model external events (e.g., rain, wind, etc) that would affect the traversability of specific regions of the state space. The evolution of the environment state $z$ may be stochastic and, most importantly, it is not perfectly observable. Thus, the controller has to make decisions based on the belief about the environment. For example, when the objective is to reach a goal location before a deadline and only partial knowledge about the traversability of the terrain is given, the controller should follow a path that maximizes the probability of reaching the goal in time, given our belief about the environment.

More formally, control actions are computed based on the environment belief vector 
$b_{k} \in \mathcal{B}=\{ b \in \mathbb{R}^{|\mathcal{Z}|} : \sum_{z=1}^{|\mathcal{Z}|} b^{(z)} = 1\}$
representing the posterior probability that the partially observable state environment $z_k$ equals $z \in \mathcal{Z}$, i.e., $b_{k} = [b_k^{(1)},\ldots, b_{k}^{(|\mathcal{Z}|)}]$ with 
\begin{equation*}
    \be^{(z)}_k \coloneqq \mathbb{P}(z_k = z | \boldsymbol{o}_k,\boldsymbol{s}_k,\boldsymbol{a}_{k-1}),~\forall z \in \{1\ldots,|\mathcal{Z}|\}.
\end{equation*} 
In the above definition, at time $k$ the observation vector $\boldsymbol{o}_k=[o_0,\ldots, o_k]$, the observable state vector $ \boldsymbol{s}_k=[s_0,\ldots, s_k]$, and the action vector $\boldsymbol{a}_{k-1}=[a_0,\ldots, a_{k-1}]$. Notice that the evolution of the environment belief vector $\be_k$ is stochastic as it is a function of the noisy observation vector $\boldsymbol{o}_k = [o_0, \ldots, o_k]$. Therefore, the planned path that maximizes the probability of completing the task should be computed online at after collecting the observation $o_k$ about the environment's state and updating the belief vector $b_k$.


\vspace{2pt}
\noindent\textbf{Synthesis Objectives: } Given the system's state $x(t) \in \mathbb{R}^{n_x}$ and $k$  observations $\boldsymbol{o}_{k-1}=[o_0, \ldots, o_{k-1}] \in \mathcal{O}^{k}$ about the environment, our goal is to design a control policy
\begin{equation}\label{eq:policyObg}
    \pi:\mathbb{R}^{n_x} \times \mathcal{O}^{k} \rightarrow \mathcal{U},
\end{equation}
which maps the state $x(t)$ and the observation vector $\boldsymbol{o}_{k-1}$ to the continuous control action $u\in\mathcal{U}$.
Furthermore, the control policy~\eqref{eq:policyObg} should guarantee that state and input constraints~\eqref{eq:lowLevelCnstr} are satisfied and that the probability of satisfying the specification~\eqref{eq:specFormula} is maximized. Notice that standard control strategies for nonlinear systems can be used to guarantee constraint satisfaction~\cite{gao2014tube, kogel2015discrete, yu2013tube, singh2017robust,kohler2020computationally,gurriet2018towards,CBF,wang2017safety}. Furthermore, standard decision making methodologies for Partially Observable Markov Decision Processes (POMDPs) can be used to synthesize a control policy which maximizes the probability of satisfying the specification~\cite{ahmadi2020stochastic,bouton2020point,haesaert2019temporal,vasile2016control,haesaert2018temporal,wang2018bounded}. In this paper, we bridge the gap between the two communities and we propose a hierarchical control scheme for nonlinear systems operating in partially observable environments, which guarantees that state and input constraints are satisfied and that the probability of satisfying the specifications is maximized.

\vspace{2pt}
\noindent\textbf{Navigation Example: } 
Figure~\ref{fig:envScheme} shows our motivating example, where a Segway has to reach a goal cell while avoiding known obstacles and exploring uncertain regions, which may be traversable with some probability.
The Segway dynamics are nonlinear and the system is open-loop unstable, for this reason it is required a low-level high frequency controller that stabilizes the system during operations.
On the other hand, at the highest level of abstraction we model the system using the discrete state $s_k \in \mathcal{S}$, which denotes the grid cell containing the nonlinear system~\eqref{eq:sysModel}, and the environment state $z_k \in \mathcal{Z}$ representing the traversability of the uncertain regions $\mathcal{R}_1$, $\mathcal{R}_2$ and $\mathcal{R}_3$. 
In this navigation example, where the traversability of $n_o$ regions is unknown, we define the vector $z_k$ as follows:
\begin{equation*}
    z_k = [z^{(1)}_k, \ldots, z^{(n_o)}_k ] \in \mathcal{Z},
\end{equation*} 
where each entry $z^{(i)}_k$ equals one if the $i$-th region is traversable and zero otherwise. For instance in the settings from Figure~\ref{fig:envScheme}, the environment's state $z_k = [0, 0, 1]$ as regions $\mathcal{R}_1$ and $\mathcal{R}_2$ are not traversable and region $\mathcal{R}_3$ is traversable.

\vspace{2pt}
\noindent\textbf{Strategy Overview}: We summarize the proposed multi-rate control architecture depicted in Figure~\ref{fig:summary}. The key idea is to divide the controller into three layers and compute the control action $u(t)$ as the summation of a high-frequency component $u_l(t)$ and a low-frequency component $u_m(t)$, i.e.,
\begin{equation*}
    u(t) = u_l(t) + u_m(t).
\end{equation*}
At the lowest level, the control action $u_l$ is updated continuously (at high frequency) and it is computed using Control Barrier Functions (CFBs), which leverage the full-nonlinear model~\eqref{eq:sysModel} to track a \textit{reference trajectory} $\bar x(t)$. 
The middle layer updates at a constant frequency the reference trajectory $\bar x(t)$ and reference input $u_m$, which is computed using a Model Predictive Controller (MPC). This reference trajectory steers the system from the current state $x(t)$ to a \textit{goal cell} $\Ggoal$.
Finally, the high-level planner computes the goal cell $\Ggoal$ based on partial observations $o_k$ about the environment.

\section{Unified Multi-Rate Architecture}\label{sec:proposedArchitecture}
In this section, we describe the multi-rate control architecture. First, we design a low-level CLF-CBF controller, which tracks a reference state-input trajectory and guarantees bounded tracking errors. Afterwards, we show how to update the state-input reference trajectory leveraging an MPC, which is designed using a goal state computed from a discrete high-level decision maker. Finally, we introduce the hierarchical multi-rate architecture, which guarantees that the synthesis objectives from Section~\ref{sec:probForm} are satisfied. 

\subsection{Low-Level Control}
We leverage CBFs and CLFs to design a low-level tracking controller for the nonlinear system~\eqref{eq:sysModel}. 
CBFs guarantee safety for nonlinear system~\cite{CBF}, but they are suboptimal as the control action is computed without forecasting the system's trajectory. For this reason, we use CBFs to enforce safety around a reference state-input trajectory that is computed at low frequency by the mid-level planner, as shown in Figure~\ref{fig:summary}. 

\vspace{2pt}
\noindent\textbf{Error Model:} 
At the lowest layer, the goal of the controller is to track a reference trajectory $\bar x(t)$. We assume that the reference trajectory is given by the following Linear Time-Varying (LTV) model:
\begin{equation}\label{eq:referenceModel}
\begin{aligned}
    \Sigma_{\bar x} : \begin{cases}
    \begin{matrix*}[l] \dot{\bar x}(t) = A_{\floor{t/T}} \bar x(t)+ B_{\floor{t/T}} u_m(t) \end{matrix*}, &  t \in \mathcal{T}\\
    \begin{matrix*}[l] \bar x^+(t) = \Delta_{\bar x}(x(t)) \end{matrix*}, &  t \in \mathcal{T}^c \\
    \end{cases},
\end{aligned}
\end{equation}
where $\mathcal{T}^c= \cup_{j=0}^\infty \{jT\}$, $\mathcal{T}= \cup_{j=0}^\infty  (jT, (j+1)T)$ and the time $T$ from~\eqref{eq:lowLevelCnstr} is specified by the user.
Furthermore, we denote $\bar x^-(t) = \lim_{\tau \shortarrow{1} t}\bar x(\tau)$ and $\bar x^+(t) = \lim_{\tau \shortarrow{7} t}\bar x(\tau)$ as the right and left limits of the reference trajectory $\bar x(t) \in \mathbb{R}^{n_x}$, which is assumed right continuous.
In the above system, the reference input $u_m(t) \in \mathbb{R}^{n_u}$ and the \textit{reset map} $\Delta_{\bar x} : \mathbb{R}^{n_x} \rightarrow \mathbb{R}^{n_x}$ maps the current state of the system $x(t)$ to the state $\bar x(t)$ of the reference trajectory. Both the reference input and the reset map are given by the middle layer as we will discuss in Section~\ref{sec:midLayer}. Finally, the time-varying matrices $(A_{\floor{t/T}}, B_{\floor{t/T}})$ are known and, in practice, may be computed linearizing the system dynamics~\eqref{eq:sysModel}, as discussed in the result section.

Given the nonlinear system~\eqref{eq:sysModel} and the LTV model~\eqref{eq:referenceModel}, we define the error state $e(t) = x(t)-\bar x(t)$ and the associated error dynamics:
\begin{equation}\label{eq:errorModel}
\begin{aligned}
    \Sigma_{e} : \begin{cases}
    \begin{matrix*}[l] \dot{e}(t) = f_e(x(t), \bar x(t), u_l(t)+u_m(t),t) \end{matrix*}, &  t \in \mathcal{T}\\
    \begin{matrix*}[l] e^+(t) = x^+(t) - \bar x^+(t)\end{matrix*}, &  t \in \mathcal{T}^c \\
    \end{cases}
\end{aligned}
\end{equation}
where the time-varying error dynamics are:
\begin{equation*}
\begin{aligned}
f_e(x, & \bar x, u_l+u_m, t)\\
&= f(x) + g(x)( u_l+u_m ) - (A_{\floor{t/T}} \bar x+ B_{\floor{t/T}} u_m).  
\end{aligned}
\end{equation*}
In the above definition, we dropped the dependence on time for states and inputs to simplify the notation. 
Furthermore, we introduce the low-level input constraint set $\mathcal{U}_l \subset \mathcal{U}$ and the mid-level input constraint set $\mathcal{U}_m \subset \mathcal{U}$ which partition the input space, i.e., 
\begin{equation*}
    \mathcal{U}_l \oplus \mathcal{U}_m = \mathcal{U}.
\end{equation*}

Next, we design a low-level controller which guarantees that the reference trajectory $\bar x(t)$ from the LTV model~\eqref{eq:referenceModel} is tracked within some error bounds.

\vspace{2pt}
\noindent\textbf{Control Barrier and Lyapunov Functions:} We show how to design a tracking controller using CBFs and CLFs~\cite{CBF}. First, we define the candidate Lyapunov function
\begin{equation}\label{eq:CLF}
    V(e) = ||e||_Q,
\end{equation}
where $||e||_Q = e^\top Q e$. Furthermore, we introduce the following safe set for the error dynamics~\eqref{eq:errorModel}:
\begin{equation}\label{eq:E}
\begin{aligned}
    & \mathcal{E} = \{ e \in \mathbb{R}^{n_x} : h_e(e) \geq 0 \} \subset \mathbb{R}^{n_x}.
\end{aligned}
\end{equation}
The above function $h_e$ is defined by the user and it depends on the application as discussed in the result section.

Finally, the CBF associated with the safe set~\eqref{eq:E}, and the CLF from~\eqref{eq:CLF} are used to define the following CLF-CBF Quadratic Program (QP):
\begin{equation}\label{eq:CLF-CBF-QP}
    \begin{aligned}
        v_l^{*}(t)= \argmin_{v_l \in \mathcal{U}_l, \gamma} ~ &||v_l||_2 + c_1 \gamma^2 \\
        \text{subject to} ~~ & \frac{\partial V(e)}{\partial e}f_e(x, \bar x, v_l+u_m) \leq - c_2V(e)+ \gamma \\
        & \frac{\partial h_{e}(e )}{\partial e}f_e(x, \bar x, v_l+u_m) \geq - \alpha(h_{e}(e)),
    \end{aligned}
\end{equation}
where we dropped the time dependence to simplify the notation, and $v_l \in \mathcal{U}_l$ is the low-level control action.
In the above QP, the parameters $c_1 \in \Rp$, $c_2\in\Rp$, and $\alpha \in \mathcal{K}^e$. Given the optimal control action $v^{*}_l(t)$ from the QP~\eqref{eq:CLF-CBF-QP}, the low-level control policy is defined as follows:
\begin{equation}\label{eq:lowLevPolicy}
    u_l(t) = \pi_{l}\big(x(t), \bar x(t), u_m(t) \big) = v_l^{*}(t).
\end{equation}

\begin{assumption}\label{ass:QPfeasibility}
The CLF-CBF QP~\eqref{eq:CLF-CBF-QP} is feasible for all $e = x-\bar x \in \mathcal{E}$ and for all $u_m \in \mathcal{U}_m$.
\end{assumption}

\begin{remark}
    We underline that Assumption~\ref{ass:QPfeasibility} is satisfied for some $\alpha \in \mathcal{K}^e$ when the set $\mathcal{E}$ is a robust control invariant for system~\eqref{eq:errorModel} with $u_m(t)\in\mathcal{U}_m$ and mild assumptions on the Lie derivative of~\eqref{eq:errorModel} hold (see \cite{CBF} for further details). The set $\mathcal{E}$ may be hard to compute and standard techniques are based on HJB reachability analysis~\cite{herbert2017fastrack}, SOS programming~\cite{singh2018robust}, Lyapunov-based methods~\cite{singh2017robust}, and Lipschitz properties of the system dynamics~\cite{chen2018data,yu2013tube}.
\end{remark}

The low-level control policy~\eqref{eq:lowLevPolicy} guarantees that the difference between the evolution of the nonlinear system~\eqref{eq:sysModel} and the LTV  model~\eqref{eq:referenceModel} is bounded. Indeed, when Assumption~\ref{ass:QPfeasibility} is satisfied, the CLF-CBF QP~\eqref{eq:CLF-CBF-QP} guarantees invariance of the safe set~\eqref{eq:E} for all $t \in (iT, (i+1)T)$ and $i\in\Zp$, as discussed in Section~\ref{sec:properties}. Next, we show how to design a mid-level planner which leverages the safe set~$\mathcal{E}$ from~\eqref{eq:E}.

\subsection{Mid-Level Planning}\label{sec:midLayer}
In this section we describe the mid-level planning strategy. 
At this level of abstraction, we assume that we are given a goal grid cell where we would like to steer the system. Afterwards, we compute a reference state-input trajectory using an MPC, which leverages a simplified model and the tracking error bounds from the previous section.

\vspace{2pt}
\noindent\textbf{Grid Model:}
Given the state $x(t) = [p^\top(t), q^\top(t)]^\top$, we define the current grid cell $\Gcurr$, which contains the nonlinear system~\eqref{eq:sysModel} for time $t \in [t^k, t^{k+1})$, i.e.,
\begin{equation}\label{eq:eqSet}
    p(t) \in \Gcurr \subset \mathcal{X}_p,~\forall t \in [t^k, t^{k+1}).
\end{equation}
Similarly, we define the goal cell $\Ggoal$, which represents the region where we want to steer the system for time $t \in [t^k, t^{k+1})$. Finally, we introduce the goal equilibrium sets $\Xcurr$ and $\Xgoal$, which collect the unforced equilibrium states that are contained into $\Gcurr$ and  $\Ggoal$, i.e., for $i\in \{\textrm{curr}, \textrm{goal}\}$
\begin{equation}\label{eq:eqSet}
\begin{aligned}
        &\mathcal{X}^k_i = \{ x =[p, q] \in \mathbb{R}^{n_x} | p \in \mathcal{C}^k_i, \dot x = f(x) = 0 \} \subset \mathbb{R}^{n_x}.
\end{aligned}
\end{equation}
Throughout this section, we assume that $t^k$, $\Xgoal$, $\Gcurr$, and $\Ggoal$ are given by the high-level planner and we synthesize a controller to drive the system from the current cell $\Gcurr$ to the goal cell~$\Ggoal$.

\vspace{2pt}
\noindent\textbf{Model Predictive Control:}
We design an MPC to compute the mid-level input $u_m(t)$ that defines the evolution of the reference trajectory~\eqref{eq:referenceModel} and to define the reset map $\Delta_{\bar x}$ for the LTV model~\eqref{eq:referenceModel}. The MPC problem is solved at $1/\textrm{T}$ Hz and therefore the reference mid-level control input is piecewise constant, i.e., 
\begin{equation*}
    \dot u_m(t) =0~\forall t \in \mathcal{T} = \cup_{k=0}^\infty (kT, (k+1)T).
\end{equation*}
Next, we introduce the following discrete time linear model:
\begin{equation}\label{eq:linearDiscreteSystem}
    \bar x^d\big((i+1)T\big) = \bar A_{i} \bar x^d\big(iT\big) + \bar B_{i} u\big(iT\big),
\end{equation}
where for all $i \in \mathbb{Z}_{0+}$
\begin{equation*}
    \bar A_{i} = e^{A_{\floor{iT/T}} T} \text{ and } \bar B_{i} = \int_0^{T} e^{A_{\floor{iT/T}}(T-\eta)}B_{\floor{iT/T}} d\eta,
\end{equation*}
for the matrices $A_{\floor{iT/T}}$ and $B_{\floor{iT/T}}$ defined in~\eqref{eq:referenceModel}. Now notice that, as the mid-level input  $u_m$ is piecewise constant, if at time $t_i = iT$ the state of the nominal model~\eqref{eq:referenceModel} $\bar x(iT) = \bar x^d(iT)$, then at time $t_{i+1} = (i+1)T$ we have that
\begin{equation}\label{eq:eqModelAcrossLayers}
\bar x^-((i+1)T)=\bar x^d((i+1)T).    
\end{equation}

Given the discrete time model~\eqref{eq:linearDiscreteSystem}, at time $t_i = iT \in \mathcal{T}^c$ we solve the following finite time  optimal control problem:
\begin{equation}\label{eq:ftocp}
\begin{aligned}
J(x(iT), N &)=\\
    \min_{\boldsymbol{v}_t, x_{i|i}^d} \quad &||x_{i|i}^d - x(iT)||_{Q_e}+ \sum_{t = i}^{i+N-1} \ell \big(  x_{t|i}^d, v_{t|i} \big) \\& \quad\quad\quad\quad\quad\quad\quad\quad\quad\quad\quad+ ||p_{i+N|i}^d-\pgoal||_{Q_f} \\
    \text{subject to} \quad & x^d_{t+1|i} = \bar A_t x^d_{t|i} + \bar B_t v^d_{t|i}\\
    &  x^d_{t|i} = \begin{bmatrix} p^d_{t+1|i} \\  q^d_{t+1|i} \end{bmatrix} \in \mathcal{X}_{p,q}^k \ominus \mathcal{E}, ~ v^d_{t|i} \in \mathcal{U}_m \\
    &  x^d_{i|i} - x(iT) \in \mathcal{E} \\
    &  x^d_{i+N|i} \in \Xgoal \ominus \mathcal{E}_p,\forall t = \{i, \ldots, i+N-1 \}
\end{aligned}
\end{equation}
where $\mathcal{E}$ is defined in~\eqref{eq:E}, $||p||_Q = p^\top Qp$,
\begin{equation}\label{eq:Xpq}
    \mathcal{X}_{p,q}^k=\bigg\{ x = \begin{bmatrix}p \\ {q} \end{bmatrix} \in \mathbb{R}^{n_x} | p \in \Gcurr \cup \Ggoal \text{ and } q \in \mathcal{X}_{q} \bigg\}
\end{equation}
and
\begin{equation}\label{eq:Ep}
    \mathcal{E}_{p}=\bigg\{ e = \begin{bmatrix}e_p \\ 0 \end{bmatrix} \in \mathbb{R}^{n_x} | \exists e_q \in \mathbb{R}^{n_p} \text{ and } \begin{bmatrix}e_p \\ e_q \end{bmatrix} \in \mathcal{E} \bigg\}.
\end{equation}
Notice that the MPC problem~\eqref{eq:ftocp} is designed based on the time-varying components $\Xgoal, \Gcurr, \Ggoal, \pgoal$ which are given by the high-level decision maker, as shown in Figure~\ref{fig:summary}.
Problem~\eqref{eq:ftocp} computes a sequence of open-loop actions $\boldsymbol{v}_t^d=[v^d_{t|t},\ldots,v^d_{t+N|t}]$ and an initial condition $x^d_{i|i}$ such that the predicted trajectory steers the system to the terminal set $\Xgoal$, while minimizing the cost and satisfying state and input constraints.
Let $\boldsymbol{v}_t^{d,*}=[v^{d,*}_{t|t},\ldots,v^{d,*}_{t+N|t}]$ be the optimal solution and $[x^{d,*}_{t|t},\ldots,x^{d,*}_{t+N|t}]$ the associated optimal trajectory, then the mid-level policy is \begin{equation}\label{eq:midLevPolicy}
\begin{aligned}
    \Pi_{m} : \begin{cases}
    \begin{matrix*}[l] {u_m}(t) =\!\pi_m\big(x(t), N\big) = v^{d,*}_{t|t} \end{matrix*} &  t \in \mathcal{T}^c\\
    \begin{matrix*}[l] \dot u_m(t) \!=\! 0 \end{matrix*} &  t \in \mathcal{T} \\
    \end{cases}.
\end{aligned}
\end{equation}
Finally, we define the reset map from the LTV model~\eqref{eq:referenceModel} as:
\begin{equation}\label{eq:returnMap}
\begin{aligned}
\Delta_{\bar x}(x(t)) = x_{t|t}^{d,*}.
\end{aligned}
\end{equation}


\begin{assumption}\label{ass:reachability}
Consider the equilibrium set $\mathcal{X}_\text{curr}^k$ defined in equation~\eqref{eq:eqSet}. For all states
$x(t) \in \Xcurr \oplus \mathcal{E} $ problem~\eqref{eq:ftocp} is feasible with horizon $N$.
\end{assumption}

The above assumption is satisfied when any equilibrium state $\bar x\in \Xcurr$ of the discrete time system~\eqref{eq:linearDiscreteSystem} can be steered to the goal equilibrium set $\Xgoal$ in at most $N$ time steps. More formally, Assumption~\ref{ass:reachability} holds when, for the discrete time system~\eqref{eq:linearDiscreteSystem}, $\Xgoal$ is $N$-step backward reachable from the set $\Xcurr$. 

In Section~\ref{sec:properties}, we will show that when the nonlinear system~\eqref{eq:sysModel} and the LTV system~\eqref{eq:referenceModel} are in closed-loop with the low-level policy~\eqref{eq:lowLevPolicy} and the mid-level policy~\eqref{eq:midLevPolicy}, then state and input constraints~\eqref{eq:lowLevelCnstr} are satisfied for system~\eqref{eq:sysModel}. 
Furthermore, the nonlinear system~\eqref{eq:sysModel} is steered from the current cell $\Gcurr$ to the goal cell $\Ggoal$ in finite time.

\begin{remark}
    We highlight that also RRT-based methods can be combined with CLF-CBF to design a multi-rate control architecture. In particular, it would be possible to leverage sampling-based methods, such as~\cite{karaman2011sampling, arslan2013use, ghosh2019kinematic, kuffner2000rrt}, to repeatedly solve online problem~\eqref{eq:ftocp}. 
    Notice that it important to consider the constraint tightening from  problem~\eqref{eq:ftocp} that accounts for the low-level tracking error. Indeed, this constraint tightening strategy allow us to guarantee safety of the nonlinear system~\eqref{eq:sysModel} in closed-loop with the proposed multi-rate control architecture, as we will discuss later on.
\end{remark}

\vspace{-0.2cm}
\subsection{High Level Decision Making}\label{sec:highLevel}
In this section, we first describe how to compute a control policy which maximizes the probability of satisfying the specifications. Afterwards, we show how to compute the time-varying components $\pgoal$, $\Gcurr$, $\Ggoal$, and $\Xgoal$ used in the MPC problem~\eqref{eq:ftocp}.

\vspace{2pt}
\noindent\textbf{Belief Model:}
For the MOMDP from Section~\ref{sec:probForm}, we introduced the belief vector $b_{k} \in \mathcal{B}$ that represents the posterior probability that the partially observable state $z_k$ equals $z \in \mathcal{Z}$.
The belief is a sufficient statistic and, for all $z'\in \mathcal{Z}$, it evolves accordingly to the following update equation: 
\begin{equation*}
\begin{aligned}
    b_{k+1}^{(z')} &=\eta O(s_{k+1},z', a_k, o_k) \\
    &\times \sum_{z\in\mathcal{Z}}T_s(s_k,z,a_k,s_{k+1})T_z(s_k,z,a_k,s_{k+1},z_{k+1})\be_k^{(z)},
\end{aligned}
\end{equation*}
where $\eta$ is a normalization constant~\cite{ong2010planning}. Notice that the above update equation can be written in a compact form, i.e.,
\begin{equation}\label{eq:beliefUpdate}
    b_{k+1} = f_b(s_{k+1}, s_k, o_k, a_k, b_k), 
\end{equation}
where $f_b: \mathcal{S}\times\mathcal{S}\times\mathcal{O}\times\mathcal{A}\times\Be \rightarrow \Be$. Finally, given the belief $b_k$, we introduce the following maximum likelihood environment state estimate:
\begin{equation}\label{eq:ML}
    \hat z_k = \argmax_{z \in \mathcal{Z}} \mathbb{P} (z_k =  z| \boldsymbol{o}_k,\boldsymbol{s}_k,\boldsymbol{a}_{k-1}) = \argmax_{z \in \mathcal{Z}} ~ b^{(z)}.
\end{equation}

\vspace{2pt}
\noindent\textbf{Control Policy: }
At the highest level of abstraction our goal is to compute a control policy $\pi_h$, which maximizes the probability that the high-level trajectory $\boldsymbol{\omega}$ satisfies the specifications $\psi$.
Such control control policy can be computed solving the following problem:
\begin{equation}\label{eq:quantitativeProblem}
    \pi_h = \argmax_{\pi} \mathbb{P}^{\pi}[\boldsymbol{\omega} \models \psi],
\end{equation}
where $\mathbb{P}^{\pi}[\boldsymbol{\omega} \models \psi]$ represents the probability that the specification $\psi$ is satisfied for the closed-loop trajectory  $\boldsymbol{\omega}$ under the policy $\pi$.
The solution to the above problem can be approximated using point-based and simulation-based strategies~\cite{bouton2020point,haesaert2019temporal,vasile2016control,haesaert2018temporal,wang2018bounded, shani2013survey}. In this work, we used the point-based strategy discussed in~\cite{rosolia2021time}. The resulting high-level control policy maps the high-level state $s_k$ and the environment belief $b_k$ to the high-level control action $a_k$, i.e., $a_k = \pi_h(s_k, b_k)$.

\begin{algorithm}[t!]
\SetAlgoLined
\textbf{inputs:} $x(t)$, $o_{k}$, $s_{k-1}$, $a_{k-1}$, $b_{k-1}$  \;
 set current high-level state $s_k = \texttt{getState}(x(t))$\! \;
 compute current set $\Gcurr = \texttt{getCell}(s_{k})$\! \;
  update belief $b_{k}$ using~\eqref{eq:beliefUpdate} \;
 compute high-level action $a_{k} = \pi_h(s_k,b_{k}) $ \;
 compute maximum likely estimate $\hat z_{k}$ using \eqref{eq:ML}\;
 update state $s_{k+1} = {f}_s(s_k, \hat z_k, a_k)$ \;
 compute goal set $\Ggoal= \! \texttt{getCell}(s_{k+1})$\! \;
 compute the forecasted action $\hat a = \pi_h(s_{k+1},b_{k})$ \;
 compute the forecasted state $\hat s_{k+2} = f_s(s_{k+1}, \hat z_k, \hat a)$ \;
 set forecasted set $\Gforc= \! \texttt{getCell}(s_{k+1})$\! \;
 get forecasted cell center $c^\textrm{forc} = \texttt{getCenter}(\Gforc)$ \;
 compute goal position $\pgoal = \text{Proj}( c^\textrm{forc}, \Ggoal)$ \;
\textbf{return:} $a_k$, $b_{k}$, $s_k$, $\Ggoal$, $\Gcurr$, $\pgoal$ \caption{Update High-Level}\label{algo:high-level_planning}
\end{algorithm}

The high-level policy~\eqref{eq:quantitativeProblem} is leveraged in Algorithm~\ref{algo:high-level_planning} to compute the goal position $\pgoal$ and the sets $\Gcurr$ and $\Ggoal$, which are used in the MPC problem~\eqref{eq:ftocp}. 
In Algorithm~\ref{algo:high-level_planning}, we first use the function \texttt{getState}, which maps the current state $x(t)$ to the high-level state $s_k$ representing  the  cell  containing the  nonlinear  system~\ref{eq:sysModel} (line~2). Then, we compute the current cell $\Gcurr$ associated with the high-level state $s_k$ using the function \texttt{getCell} (line~$3$). Afterwards, we update the belief state $b_k$ and we compute the control action $a_k$ (lines $4-5$). Given the control action $a_k$ and the maximum likelihood estimator of the environment state $\hat z_k$, we update the high-level state and we compute the goal cell $\Ggoal$ (lines $6-8$). Next, given the current belief $b_k$, we compute the action $\hat a$ that the high-level planner would select at the next update $k+1$ assuming that the new observation is not informative, i.e., the belief $b_{k+1} = b_k$. We leverage the action $\hat a$ to estimate the high-level state $\hat{s}_{k+2}$, which represents the location where we should steer the system after transitioning to the high-level state $s_{k+1}$, if $b_{k+1}=b_k$. The state  $\hat{s}_{k+2}$ is used to incorporate forcast into the high-level planner. In particular, given the $\hat{s}_{k+2}$ we compute the forcasted cell center $c^\textrm{forc}\in \mathbb{R}^{n_p}$ and the forecasted cell $\Gforc$ representing the grid cell where the high-level decision maker would like to steer the system, if no informative observations are collected (lines~$11-12$). Finally, the goal cell $\Ggoal$ and the forecasted center $c^\textrm{forc}\in\mathbb{R}^{n_p}$ are used to compute the goal position $p^\mathrm{goal}$~(line~$13$).

\begin{figure}[t!]
    \centering
	\includegraphics[trim= 3mm 3mm 3mm 3mm, clip, width= 0.99\columnwidth]{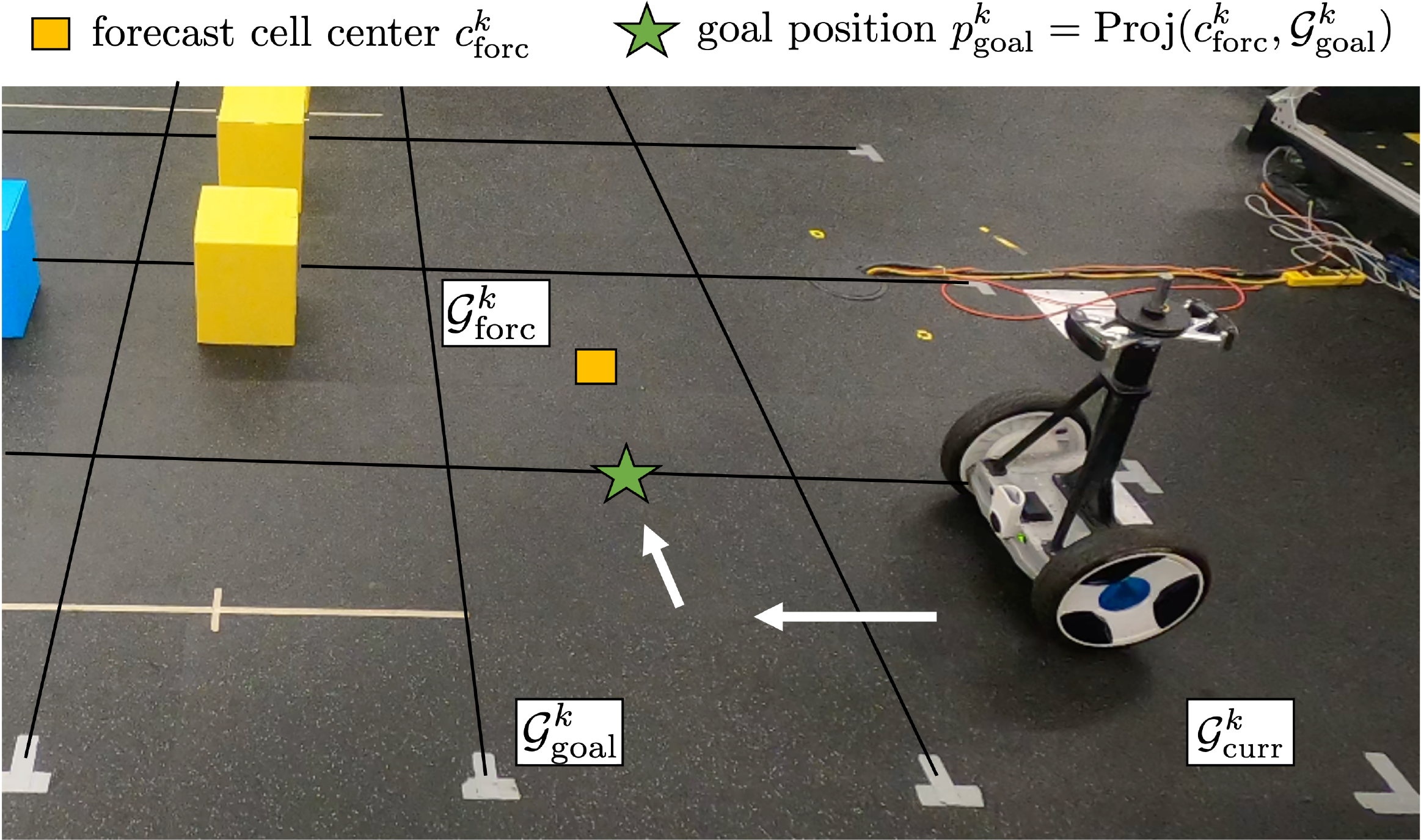}
    \caption{The above figure illustrates the high-level updated from Algorithm~\ref{algo:high-level_planning} that is used to compute the goal position (green star).}
    \label{fig:visAlgo}
\end{figure}

Figure~\ref{fig:visAlgo} illustrates the high-level update from Algorithm~\ref{algo:high-level_planning} that is used to compute the goal position leveraged in the design of the mid-level MPC. In this example, the Segway is located in the bottom right corner of the grid and the current high-level action $a_k$ is to move west. The figure shows also the forecasted action $\hat a$ that the Segway would take from the goal region, if the belief $b_k$ is not updated. Basically, $\hat a$ is a high-level open-loop prediction of the future control action and it is used to incorporate forecast into the high-level decision maker. Indeed, the goal position $\pgoal$ is computed projecting the forecasted cell center $c^\mathrm{forc}$ onto the goal cell $\Ggoal$.

\begin{algorithm}[t!]
\SetAlgoLined
\textbf{inputs:} $k$, $s_k$, $b_k$, $a_k$, $i$, $x(t)$, $u_m(t)$, $\bar x(t)$,  $\Gcurr$, $\Ggoal$, $\pgoal$, $N^k_i$, $\GcurrMinus$, $\GcurrMinus$, $\pgoalMinus$, $N^{k-1}_i$ \;
\If{$p(t) \in \Ggoal$ or $k=0$} { 
\tcp{Update high-level goal}
measure $o_{k+1}$ \;
update $a_{k+1}$,$b_{k+1}$, $s_{k+1}$, $\Xgoal$, $\GgoalPlus$, $\GcurrPlus$, $\pgoalPlus$ using Algorithm~\ref{algo:high-level_planning} with $x(t)$, $o_{k+1}$, $s_{k}$, $a_{k}$, $b_{k}$ \;
set $N^{k+1}_i = N$ \;
$k = k + 1$ \;
}
\If{$t \in \mathcal{T}^c = \cup_{j=0}^\infty \{jT\}$}{ 
    \tcp{Update mid-level plan}
    solve MPC problem~\eqref{eq:ftocp} with $N=N^k_i$, and $\Xgoal$, $\Gcurr$, $\Ggoal$, $\pgoal$ \;
    \eIf{the MPC problem~\eqref{eq:ftocp} is not feasible}{ 
        solve MPC problem~\eqref{eq:ftocp} with $N=N^{k-1}_i$, and  $\XgoalMinus$, $\GcurrMinus$, $\GgoalMinus$, $\pgoalMinus$ \;
        set $N^{k-1}_{i+1} = \max(1, N^{k-1}_{i} - 1)$ \;
        set $N^{k}_{i+1} = N^{k}_{i}$ \;
    }{
        set $N^{k-1}_{i+1} = N^{k-1}_{i}$ \;
        set $N^{k}_{i+1} = \max(1, N^{k}_{i} - 1)$ \;
    }
    set $u_m(t)= v_{t|t}^{d,*}+K(x(t)- \bar x_{t|t}^{d,*})$ \;
    update $\bar x(t)= \Delta_{\bar x}(x(t)) = \bar x_{t|t}^{d,*}$ \;
    $i = i+1$ \;
}
\tcp{Compute low-level control}
solve the CBF problem~\eqref{eq:CLF-CBF-QP} \;
Compute total input $u(t) = u_l(t)+u_m(t)$ \;
\textbf{Return:}{ $u(t)$, $k$, $s_k$, $b_k$, $a_k$, $i$, $x(t)$, $u_m(t)$, $\bar x(t)$, $\Gcurr$, $\Ggoal$, $\pgoal$, $N^k_i$, $N^{k-1}_i$}
\caption{Multi-Rate Control}\label{algo:multiRate}
\end{algorithm}

\vspace{-0.3cm}
\subsection{Control Architecture}
Finally, we introduce the multi-rate hierarchical control architecture which leverages the low-level, mid-level, and high-level control policies from the previous sections.
The multi-rate control Algorithm~\ref{algo:multiRate} describes the architecture depicted in Figure~\ref{fig:summary}. When the nonlinear system~\eqref{eq:sysModel} reaches the goal cell (i.e., $p(t) \in \Ggoal$), the high-level decision maker reads the new observations $o_{k+1}$ and updates high-level state, action, goal position $\pgoal$, goal cell $\Ggoal$, and current cell $\Gcurr$ (lines~$3-4$). Finally, it updates the high-level time $k$ and it initializes the MPC horizon $N_i^k=N$. 
Afterwards, the mid-level planner (lines $8-20$) updates the mid-level time counter $i$ and the planned trajectory at a constant frequency of $1/T$ Hz. First, it solves the MPC problem~\eqref{eq:ftocp} with $N=N_i^k$ and time-varying components $\Xgoal$, $\Ggoal, \Gcurr$, and $\pgoal$. If the MPC problem is not feasible, the planner computes a contingency plan (lines $10-14$), otherwise it updates the prediction horizon (lines~15-16). Note that the MPC problem solved in line~9 of Algorithm~\ref{algo:high-level_planning} may be not feasible as the terminal constraint set $\Xgoal$ is updated by the high-level planner. For this reason, we introduced the contingency plan (lines $10-14$), where the MPC problem from line~11 is constructed using the terminal constraint set $\XgoalMinus$. As we will show in the proof of Theorem~\ref{th:safe}, when the MPC problem constructed with terminal constraint set $\Xgoal$ is not feasible, we can guarantee the feasibility of the contingency MPC with $\XgoalMinus$ as terminal constraint. This fact allows us to guarantee safety for the closed-loop system.
Finally, Algorithm~\ref{algo:multiRate} computes the low-level control action solving the CLF-CBF~QP~\eqref{eq:CLF-CBF-QP} and the total control input that is given by the summation of the mid-level and low-level control actions, i.e.,
\begin{equation*}
    u(t) = u_l(t)+u_m(t).
\end{equation*}

\section{Safety and Performance Guarantees}\label{sec:properties}
In this section we show the properties of the proposed multi-rate control architecture. We consider the augmented system:
\begin{equation}\label{eq:augSys}
\begin{aligned}
    \Sigma_{\textrm{aug}} : \begin{cases}
    \begin{matrix*}[l]\dot x(t) = f\big(x(t)\big) +g\big(x(t)\big)\big(u_l(t)+u_m(t)\big), & t\geq0 \\
    \dot{\bar x}(t) = A_{\floor{t/T}} \bar x(t)+ B_{\floor{t/T}} u_m(t) , &  t \in \mathcal{T}\\
    \bar x^+(t) = \Delta_{\bar x}(x(t)) , &  t \in \mathcal{T}^c \\
    \end{matrix*}
    \end{cases}
\end{aligned}
\end{equation}
where the nonlinear dynamics for state $x(t) \in \mathbb{R}^{n_x}$ are defined in~\eqref{eq:sysModel} and the LTV model for the nominal state $\bar x(t) \in \mathbb{R}^{n_x}$ is defined in~\eqref{eq:referenceModel} for the reset map~\eqref{eq:returnMap} given by the MPC. 
In what follows, we analyse the properties of the proposed multi-rate control Algorithm~\ref{algo:multiRate} in closed-loop with system~\eqref{eq:augSys}. 
We show that the closed-loop system satisfies state and input constraints~\eqref{eq:lowLevelCnstr} and that the proposed algorithm maximizes the probability of satisfying the specifications. Notice that in practice the state $x(t)$ is given by the nonlinear system~\eqref{eq:sysModel}, whereas the nominal state $\bar x(t)$ is computed by the low-level layer to update the tracking error $e(t)$, as shown in Figure~\ref{fig:summary}.

\begin{proposition}\label{Proposition:llProp}
Consider the closed-loop system~\eqref{eq:lowLevPolicy} and~\eqref{eq:augSys} with mid-level input $u_m(t)\in\mathcal{U}_m$ and $\dot u_m(t)=0,\forall t \in \mathcal{T}$. If Assumption~\ref{ass:QPfeasibility} holds and the error $e(kT) = x(kT) - \bar x(kT) \in \mathcal{E}$ for all $k \in \Zp$, then the control policy~\eqref{eq:lowLevPolicy} guarantees that $e(t) \in \mathcal{E}$ and $u_l(t)\in \mathcal{U}_l,~\forall t\in [kT, (k+1)T)$.
\end{proposition}
\begin{proof}
The proof follows from standard CBF arguments~\cite{CBF}. First, we notice that the error $e(kT)= x(kT) - \bar x(kT)$ follows the error dynamics in~\eqref{eq:errorModel}. Furthermore, by construction the time-varying matrices $(A_{\floor{t/T}}, B_{\floor{t/T}})$ are constant for $t \in [kT, (k+1)T)$. Therefore, for all $k\in\Zp$ and $t \in [kT, (k+1)T)$, we have that error dynamics in~\eqref{eq:errorModel} are nonlinear control affine for the low-level input $u_l$. This fact implies that, if at time $t=kT$ the error $e(kT) = x(kT) - \bar x(kT) \in \mathcal{E}$, then from the feasibility of the CLF-CBF QP~\eqref{eq:CLF-CBF-QP} from Assumption~\ref{ass:QPfeasibility} we have that $e(t)=x(t)-\bar x(t) \in \mathcal{E},~\forall t\in [kT, (k+1)T)$.
\end{proof}

Proposition~\ref{Proposition:llProp} shows that between time $t_i=iT$ and $t_{i+1}=(i+1)T$ the difference between the state $x$ and the state $\bar x$ of the reference trajectory is bounded. Next, we show that this property allows us to guarantee safety and convergence in finite time of the nonlinear system~\eqref{eq:sysModel} to a goal cell $\Ggoal$ contained in the feasible region. In turns, convergence in finite time allows us to show that the proposed approach maximizes the probability of satisfying the high-level specifications.

\begin{assumption}\label{ass:freePath}
Algorithm~\ref{algo:high-level_planning} returns a goal cell $\Ggoal$ which is contained in the feasible set $\mathcal{X}_p$, i.e., $\Ggoal\subset\mathcal{X}_p$.
\end{assumption}


\begin{theorem}\label{th:safe}
Let Assumptions~\ref{ass:QPfeasibility}-\ref{ass:freePath} hold and consider system~\eqref{eq:augSys} in closed-loop with Algorithm~\ref{algo:multiRate}. If  at time $t_i=iT$ the MPC problem~\eqref{eq:ftocp} is feasible with $N_i^k = N$ and time-varying components $\Xgoal, \Gcurr, \Ggoal$, and $\pgoal$, then there exists a $j\in\{i,\ldots, i+N-1\}$ such that the closed-loop system satisfies state and input constraints~\eqref{eq:lowLevelCnstr} for all 
$ t \in \{iT,\ldots, jT\}$  and the state $x((j+1)T) = [p^\top((j+1)T), q^\top((j+1)T)]^\top$ reaches the goal cell $\Ggoal$, i.e., $p((j+1)T) \in \Ggoal$.
\end{theorem}
\begin{proof}
From Assumption~\ref{ass:freePath} we have that
the high-level policy~\eqref{eq:quantitativeProblem}
takes a high-level action $a_k$ which avoids collision with the obstacles, i.e.,
\begin{equation}\label{eq:feasibleGoal}
    \Ggoal\subset\mathcal{X}_p.
\end{equation}

Next, we show that if at time $t_i = iT$ the MPC problem~\eqref{eq:ftocp} is feasible with $\Xgoal, \Ggoal, \Gcurr, \pgoal$, and $N_i^k>1$, then at time $t_{i+1}=(i+1)T$ the MPC problem~\eqref{eq:ftocp} is feasible with $\Xgoal, \Ggoal, \Gcurr, \pgoal$, and $N^k_{i+1}=N^k_i-1$. Let 
\begin{equation}\label{eq:optPlannerTrajectory}
    [x_{i|i}^{d,*}, x_{i+1|i}^{d,*},\ldots,x_{i+N^k_i|i}^{d,*}] \text{ and } [u_{i|i}^{d,*},\ldots,u_{i+N^k_i-1|i}^{d,*}]
\end{equation}
be the optimal state input sequence to the MPC problem~\eqref{eq:ftocp} at time $t_i = iT$. Then, from Proposition~\ref{Proposition:llProp}, equation~\eqref{eq:eqModelAcrossLayers}, and the definition of the reset map~\eqref{eq:returnMap}, we have that 
\begin{equation}\label{eq:discContConn}
    x((i+1)T)-\bar x_{i+1|i}^{d,*} = x((i+1)T)-\bar x^-((i+1)T) \in \mathcal{E},
\end{equation}
and therefore, by feasibility of~\eqref{eq:optPlannerTrajectory} at time $t_i$, the following sequences of $N^k_i-1$ states and $N^k_i-2$ inputs
\begin{equation}\label{eq:feasTr}
    [x_{i+1|i}^{d,*},\ldots,x_{i+N^k_i|i}^{d,*}] \text{ and } [u_{i+1|i}^{d,*},\ldots,u_{i+N^k_i-1|i}^{d,*}]
\end{equation}
are feasible at time $t_{i+1}=(i+1)T$ for the MPC problem~\eqref{eq:ftocp} with $\Xgoal, \Ggoal, \Gcurr, \pgoal$, and $N^k_{i+1}=N^k_i-1$. 


Now, we show that state and input constraints are satisfied until the system reaches the goal set $\Ggoal$. 
Recall that by assumption the MPC problem is feasible at time $t_i=iT$ with $\Xgoal, \Ggoal, \Gcurr, \pgoal, N_i=N$ and assume that $p(jT) \notin \Ggoal$ for all $j \in \{i,\ldots, i+N-1\}$. By induction the MPC problem~\eqref{eq:ftocp} with $\Xgoal, \Ggoal, \Gcurr, \pgoal$ and $N^k_j=N^k_i-j$ is feasible for all $j\in\{i,\ldots, i+N-1\}$. Consequently, Algorithms~\ref{algo:high-level_planning} 
returns a feasible mid-level control action\footnote{Note that as $p(jT) \notin \Ggoal$ for all $j \in \{i,\ldots, i+N-1\}$ the MPC time-varying components are not updated.} $u_m(t)\in\mathcal{U}_m$.
Furthermore, from Proposition~\ref{Proposition:llProp} we have that the low-level controller returns a feasible control action $u_l(t)\in\mathcal{U}_l$ and therefore
\begin{equation}
    u(t) = u_l(t) + u_m(t) \in \mathcal{U}_l \oplus \mathcal{U}_m = \mathcal{U}, \forall t \in \Rp.
\end{equation}
The feasibility of the state-input sequences in~\eqref{eq:feasTr} for the MPC problem solved with $\Xgoal, \Ggoal, \Gcurr, \pgoal$ implies that
\begin{equation}\label{eq:feasSolCnstr}
    \begin{aligned}
    & x_{j|j}^{d,*} \in \mathcal{X}_{p,q}^k \ominus \mathcal{E} \\
    & x(jT) - x_{j|j}^{d,*} \in \mathcal{E},
    \end{aligned}
\end{equation}
$\forall j\in\{i,\ldots, i+N-1\}$. Consequently, from the above equation and definition~\eqref{eq:Xpq}, we have that 
\begin{equation*}
    \begin{aligned}
        & p(jT) \in \mathcal{X}_{p} \text{ and } q(jT) \in \mathcal{X}_{q}, \forall j\in\{i,\ldots, i+N-1\}. 
    \end{aligned}
\end{equation*}

Finally, we show that the state $x(t)$ of the augmented system~\eqref{eq:augSys} in closed-loop with Algorithm~\ref{algo:multiRate} converges to the goal cell $\Ggoal$ in finite time. 
We have shown that, if $p(jT) \notin \Ggoal$ for all $j\in\{i,\ldots,i+N-1\}$, then the MPC problem is feasible for all time $t_k=kT$ and $k \in \{i, \ldots,i+N-1\}$. Now we notice that by feasibility of the MPC problem at time $t_{i+N-1}=(i+N-1)T$ with $N_{i+N-1}=1$, we have that the optimal planned trajectory satisfies
\begin{equation}\label{eq:optPlannedLast}
    x_{i+N|i+N-1}^{d,*} \in \Xgoal \ominus \mathcal{E}_p,
\end{equation}
where $\mathcal{E}_p$ is defined as in~\eqref{eq:Ep}.

From equation~\eqref{eq:eqModelAcrossLayers} and Proposition~\ref{Proposition:llProp}, we have that 
\begin{equation*}
    x((i+N)T)-\bar x_{i+N|i+N-1}^{d,*} \!\!=\! x((i+N)T)-\bar x^-((i+N)T) \in \mathcal{E}.
\end{equation*}
The above equation together with~\eqref{eq:Ep} and~\eqref{eq:optPlannedLast} imply that at time $t_{i+N}=(i+N)T$
\begin{equation*}
\begin{aligned}
    x((i+N)T) = \begin{bmatrix} p((i+N)T) \\ q((i+N)T)\end{bmatrix}& \in \Xgoal \ominus \mathcal{E}_p \oplus \mathcal{E}
\end{aligned}
\end{equation*}
and therefore $p((i+N)T) \in \Ggoal$.

Concluding, if for all time $t_j=jT$ and $j \in \{i, \ldots, i+N-1\}$ we have that $p(jT) \notin \Ggoal$, then $p((i+N)T) \in \Ggoal$. Thus, the closed-loop system converges to the goal cell $\Ggoal$ in finite time.
\end{proof}

\vspace{0.2cm}
Finally, we leverage Theorem~\ref{th:safe} to show that, when $\Ggoal \subset \mathcal{X}_p$, the multi-rate control Algorithm~\ref{algo:multiRate} steers the system in finite time to goal cell $\Ggoal$ for all $k \in \Zp$ and, consequently, the closed-loop system maximizes the probability of satisfying the high-level specifications. In particular, we show that the contingency plan from lines~10--14 of Algorithm~\ref{algo:multiRate} guarantees feasibility of the planner when the time-varying components are updated.

\begin{theorem}\label{th:spec}
Let Assumptions~\ref{ass:QPfeasibility}-\ref{ass:freePath} hold and consider system~\eqref{eq:augSys} in closed-loop with Algorithm~\ref{algo:multiRate}. If $x(0) \in \Xcurr\oplus \mathcal{E}$, then Algorithm~\eqref{algo:multiRate} maximizes the probability that the closed-loop satisfies the high-level specifications.
\end{theorem}
\begin{proof}
The proof follows by induction. Assume that at time $t_{i}=iT$ the closed-loop system reaches the goal cell $\Ggoal$, i.e., $p(iT)\in\Ggoal$. Then, at time $t_{i}=iT$ we have that the high-level decision maker from Algorithm~\ref{algo:multiRate} (lines~2--9) updates the high-level time and the time-varying components $\XgoalPlus, \GcurrPlus, \GgoalPlus, \pgoalPlus$ used to design the MPC problem~\eqref{eq:ftocp}. After the high-level update, the MPC problem with $N=N^{k+1}_j$, $\XgoalPlus, \GcurrPlus, \GgoalPlus, $ and $\pgoalPlus$ may be either feasible or unfeasible\footnote{Unfeasiblity may be caused by the update of $\GgoalPlus$, $\GcurrPlus$ and $\XcurrPlus$.}. Thus, we analyse the following three cases for $j\geq i$:

\textit{Case 1: } The MPC problem with $\GgoalPlus, \GcurrPlus, \pgoalPlus$, and $N=N^{k+1}_j$ is feasible, therefore from Theorem~\ref{th:safe} we have that Algorithm~\ref{algo:multiRate} steers the nonlinear system to the goal $\GgoalPlus$.

\textit{Case 2:}
The MPC problem with $\GgoalPlus, \GcurrPlus, \pgoalPlus$, and $N=N^{k+1}_j$ is not feasible and $N_j^k = 1$. Then from Theorem~\ref{th:safe}, we have that the contingency MPC with $N_j^k$, $\Ggoal, \Gcurr$ and $\pgoal$ is feasible and Algorithm~\ref{algo:high-level_planning} returns a feasible control action. Furthermore, as $N_j^k = 1$ the terminal state of the optimal predicted trajectory is
\begin{equation*}
    x_{j+1|j}^{d,*} \in \Xgoal \ominus \mathcal{E}_p.
\end{equation*}
The above equation together with equation~\eqref{eq:discContConn} imply that
\begin{equation*}
    x((j+1)T) \in \Xgoal \ominus \mathcal{E}_p \oplus \mathcal{E} \subset \Xgoal \oplus \mathcal{E}.
\end{equation*}
Notice that $\Xgoal = \XcurrPlus$, thus from Assumption~\ref{ass:reachability} we have that at the next time step  $t_{j+1}=(j+1)T$ the MPC problem with $N^{k+1}_{j+1}=N$, $\XgoalPlus$, $\GgoalPlus$, $\GcurrPlus$, and $\pgoalPlus$ is feasible and, from Theorem~\ref{th:safe}, we have that Algorithm~\ref{algo:multiRate} steers the nonlinear system to the goal $\GgoalPlus$ in finite time.

\textit{Case 3:}
The MPC problem with $\GgoalPlus, \GcurrPlus, \pgoalPlus$, and $N=N^{k+1}_j$ is not feasible and $N_j^k > 1$. Then from Theorem~\ref{th:safe}, we have that the contingency MPC with $N_j^k$, $\Ggoal, \Gcurr$ and $\pgoal$ is feasible (lines~10--13 in Algorithm~\ref{algo:multiRate}). 

By assumption $x(0)  \in \Xcurr \oplus \mathcal{E}$, thus from Assumption~\ref{ass:reachability} and Theorem~\ref{th:safe} we have that at time $t=0$ the MPC is feasible and that Algorithm~\ref{algo:multiRate} steers system~\eqref{eq:sysModel} to $\mathcal{G}^0_{\mathrm{goal}}$. 
This fact implies that the conditions form one the above Cases~1--3 are met, if $x(0)  \in \Xcurr \oplus \mathcal{E}$. Now notice that at each time step $N_{j+1}^k = N_{j}^k - 1$ (line 16), thus  Case~3 occurs at most $N$ times. Therefore, after at most $N$ time steps the conditions from either Case~1 or Case~2 are met and Algorithm~\ref{algo:multiRate} will steer the system to the grid cell $\mathcal{C}_{\textrm{goal}}^{k+1}$ computed by the high-level policy~\eqref{eq:quantitativeProblem}.
Consequently, as the high-level policy~\eqref{eq:quantitativeProblem} maximizes the probability of satisfying the specifications, we have that the closed-loop system maximizes the probability of satisfying the specifications.
\end{proof}

\begin{figure}[b!]
    \centering
	\includegraphics[trim= 3mm 3mm 3mm 3mm, clip, width= 0.99\columnwidth]{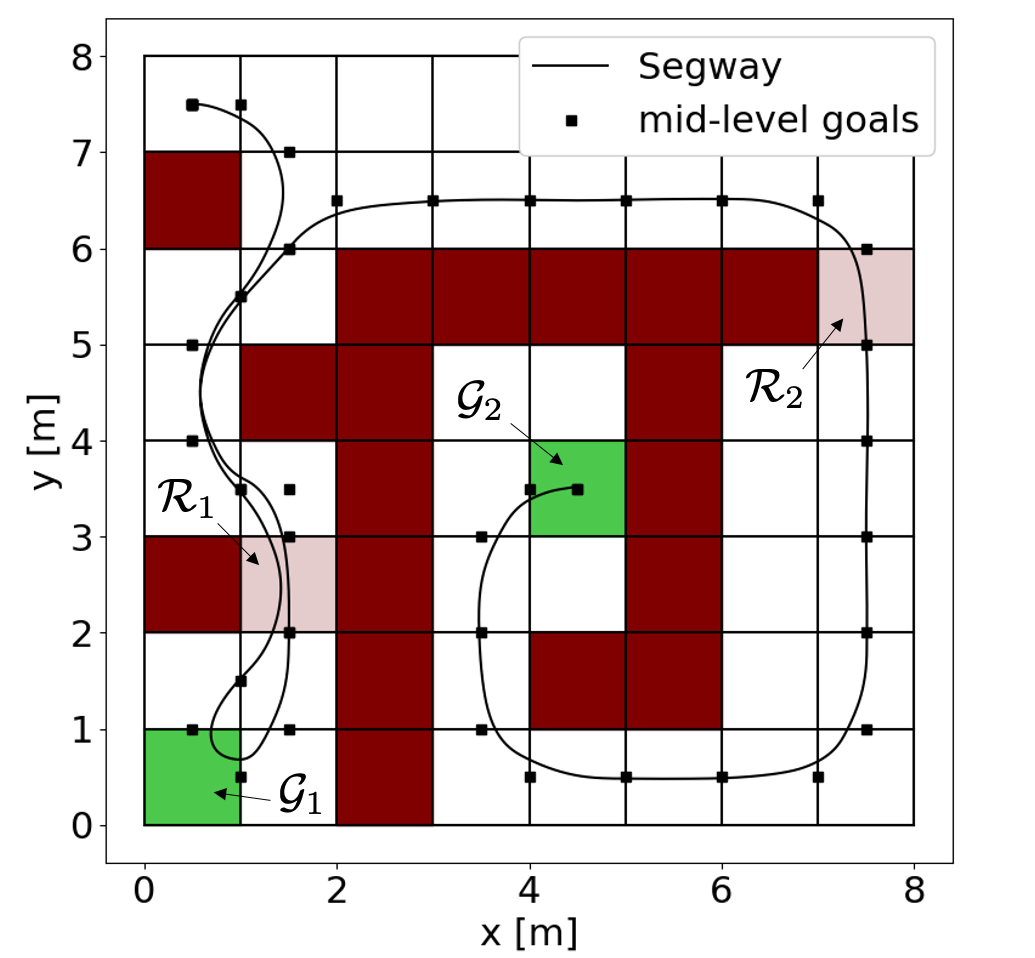}
    \caption{Closed-loop trajectory. The Segway first explores regions $\mathcal{R}_1$, which is traversable, and $\mathcal{G}_1$ that does not contain the science sample. Afterwards, it explores the traversable region $\mathcal{R}_2$ and it reaches $\mathcal{G}_2$.}
    \label{fig:clSim}
\end{figure}

\section{Results}
We tested the proposed strategy in simulation and experiment on navigation tasks inspired by the Mars exploration mission~\cite{haesaert2018temporal,haesaert2019temporal,nilsson2018toward}. 
We control a Segway-like robot and our goal is to explore the environment to find science samples which may be located in the known goal regions $\mathcal{G}_i$ shown in Figure~\ref{fig:clSim}. 
The specification $\psi = \neg \texttt{collision} U (( \texttt{Goal}_1 \land \texttt{sample}_1) \lor ( \texttt{Goal}_2 \land \texttt{sample}_2) )$, where the atomic proposition $\texttt{sample}_i$ is satisfied if the region $\mathcal{G}_i$ contains a science sample and the atomic proposition $\texttt{Goal}_i$ is satisfied if the Segway is in a goal cell $\mathcal{G}_i$. The high-level control policy associated with specification $\psi$ is computed solving a reach-avoid problem for the product MOMDP,  which is computed preforming the cross-product between an automata associated the specification $\psi$ and the original MOMDP\footnote{The computational complexity of solving the high-level synthesis problem is a function of the dimension of the product MOMDP, which may grow exponentially for complex specifications. The analysis of the computational tractability of the high-level synthesis process is beyond the scope of this work and the code used to synthesize the high-level policy can be found at \texttt{\url{https://github.com/urosolia/MOMDP}}.}. For further details on how to convert a specification into a finite state automata and the computation of the product, please refer to~\cite{belta2017formal}.
While performing the search task, we have to collect observations to determine the state of the uncertain region $\mathcal{R}_i$, which may be traversable with some probability.  The controller has access only to partial observations about the environment.
In particular, the Segway receives a perfect observation about the state of the uncertain region $\mathcal{R}_i$ when one cell away, an observation which is correct with probability $0.8$, when the Manhattan distance is smaller than two, and an uninformative observations otherwise. Similarly, the Segway receives a partial noisy observation about the goal region $\mathcal{G}_i$ which is correct with probability $0.7$, when one cell away and a perfect observations when the goal cell $\mathcal{G}_i$ is reached.

\begin{figure}[t!]
    \centering
	\includegraphics[trim= 3mm 3mm 3mm 3mm, clip, width= 1.0\columnwidth]{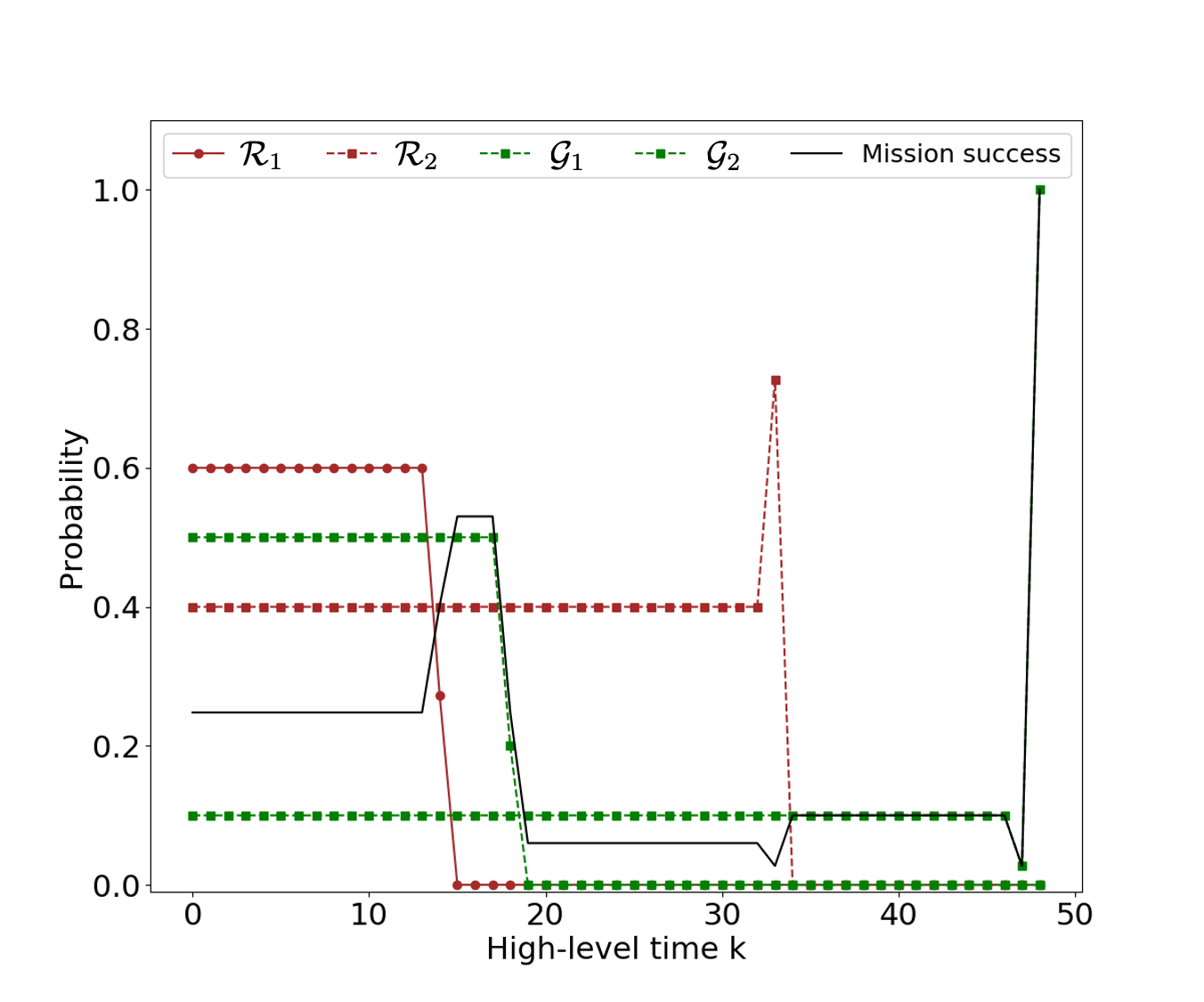}
    \caption{This figure shows the closed-loop probability of mission success, which equals the probability of satisfying the high-level specifications. Furthermore, we reported also the belief about regions $\mathcal{R}_1$ and $\mathcal{R}_2$ being traversable (brown) and the goal regions $\mathcal{G}_1$ and $\mathcal{G}_2$ containing the science sample (green).}
    \label{fig:prob}
\end{figure}

\begin{figure}[t!]
    \centering
	\includegraphics[trim= 3mm 3mm 3mm 3mm, clip, width= 1.0\columnwidth]{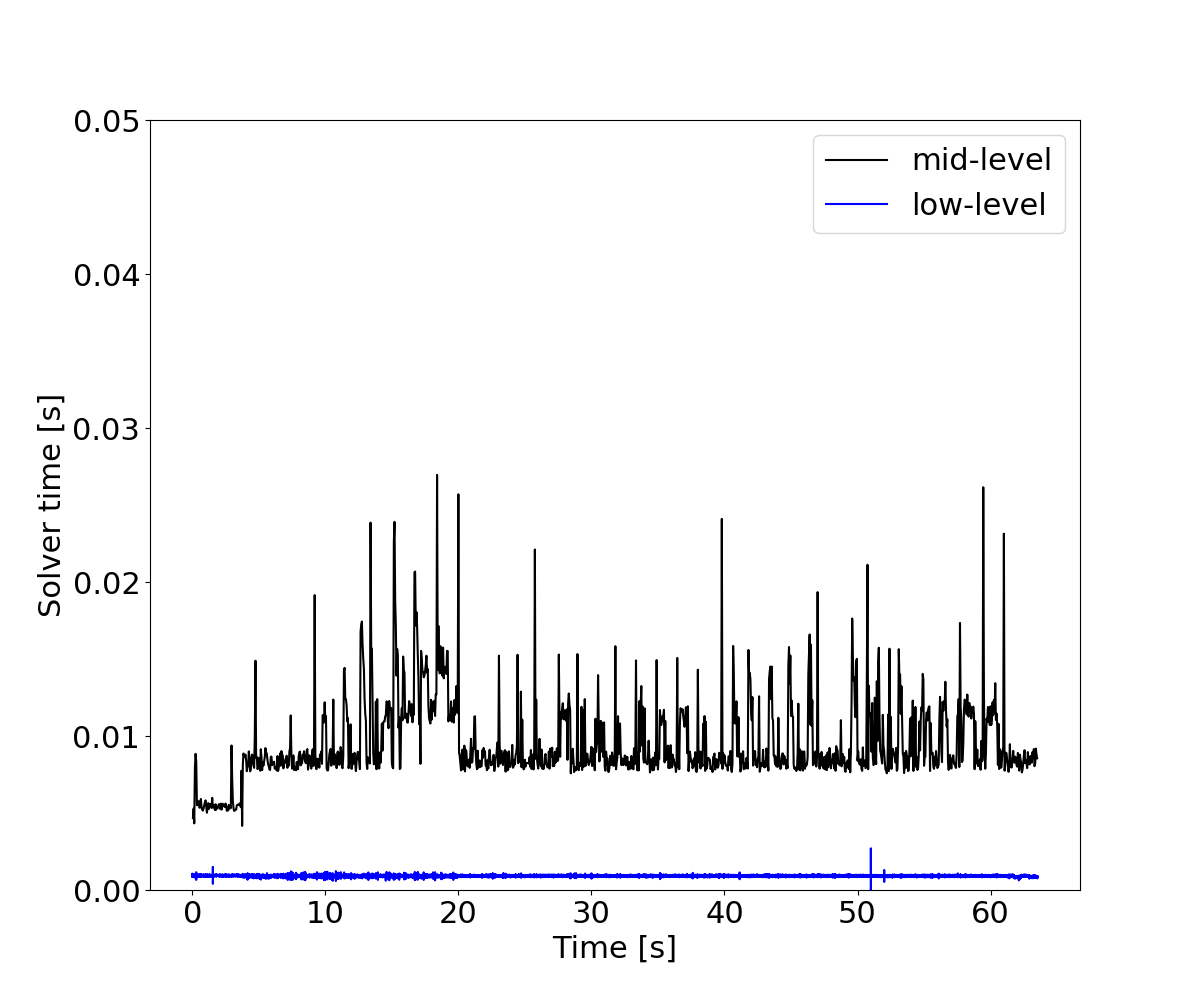}
    \caption{Computational time associated with the mid and low layers. It takes on average $12$ ms to compute mid-level control actions and less than $1$ ms to compute low-level commands. In this example, the mid-level is discretized at $20$ Hz and the low-level at $1$ kHz.}
    \label{fig:comp}
\end{figure}

 The state of the Segway is defined as follows:
 \begin{equation*}
     x= [X, Y, \theta, v, \dot \theta, \psi, \dot \psi],
 \end{equation*}
 where $(X-Y)$ represents the position of the center of mass, $(\theta, \dot \theta)$ the heading angle and yaw rate, $v$ the velocity, and $(\psi, \dot \psi)$ the rod's angle and angular velocity. The control input $u = [T_l, T_r]$, where $T_l$ and $T_r$ are the torques to the left and right wheel motors, respectively. In order to implement the low-level CLF-CBF QP we used the following function:
\begin{equation}\label{eq:barrier}
    h(e) = 1-|| \text{diag}(v_h) (x-\bar x) ||_2^2,
\end{equation}
where $v_h = [ 1/0.02, 1/0.02, 1/0.1, 1/0.1, 1/0.3, 1/0.1, 1/0.3]$ and  $\bar x = [\bar X, \bar Y, \bar \theta, \bar v, \dot{\bar \theta}, \bar \psi, \dot{\bar \psi}]$ represents the state of the nominal system from~\eqref{eq:referenceModel}. The candidate control Lyapunov function is
\begin{equation*}
    V(e) = || \text{diag}(v_v) (x-\bar x) ||_2^2,
\end{equation*}
where  $v_v = [ 100,100,100,100,10000,10000,100]$ and in the CLF-QBF QP~\eqref{eq:CLF-CBF-QP} we used $c_1=1$, $c_2=10$ and $\alpha(x)=x$.
The planning model~\eqref{eq:referenceModel} is computed iteratively linearizing the Segway dynamics around the predicted MPC trajectory. This strategy is standard in MPC, for more details on the linearization strategy please refer to~\cite{rosolia2019learning}.  The stage cost 
\begin{equation*}
\ell(x,u)=x^\top Qu+u^\top Ru,
\end{equation*}
for the tuning matrices $Q=\text{dial}(0.1, 0.1, 0, 0, 10, 1, 10)$, $R=\text{diag}(0.01, 0.01)$, and $Q_f = \text{diag}(100, 100)$. Furthermore, we added an input rate cost with penalty $Q_{\text{rate}}=0.1$ and a slack variable for the terminal constraint on the state $q_{t+N|i}$ with weight $Q_{\text{slack}}=\text{diag}(100, 100, 100)$. Finally, we approximated $\mathcal{S}= \{e=x-\bar x \in\mathbb{R}^{n_x}: h(e) \geq 0\}= \{e = x-\bar x \in\mathbb{R}^{n_x}: || \text{diag}(v_v) (x-\bar x) ||_2^2 \leq 1\} $ with $\bar{\mathcal{S}}= \{e = x-\bar \in\mathbb{R}^{n_x}: || \text{diag}(v_v) (x-\bar x) ||_\infty \leq 1\} $. This strategy allows us to write the MPC problem~\eqref{eq:ftocp}
as a QP\footnote{Note that using $\mathcal{S}$ renders the MPC problem an SOCP, which is convex but computationally more demanding.}, which we solved using OSQP~\cite{stellato2020osqp, banjac2017embedded}.

\begin{figure}[t!]
    \centering
	\includegraphics[trim= 3mm 3mm 3mm 3mm, clip, width= 1.0\columnwidth]{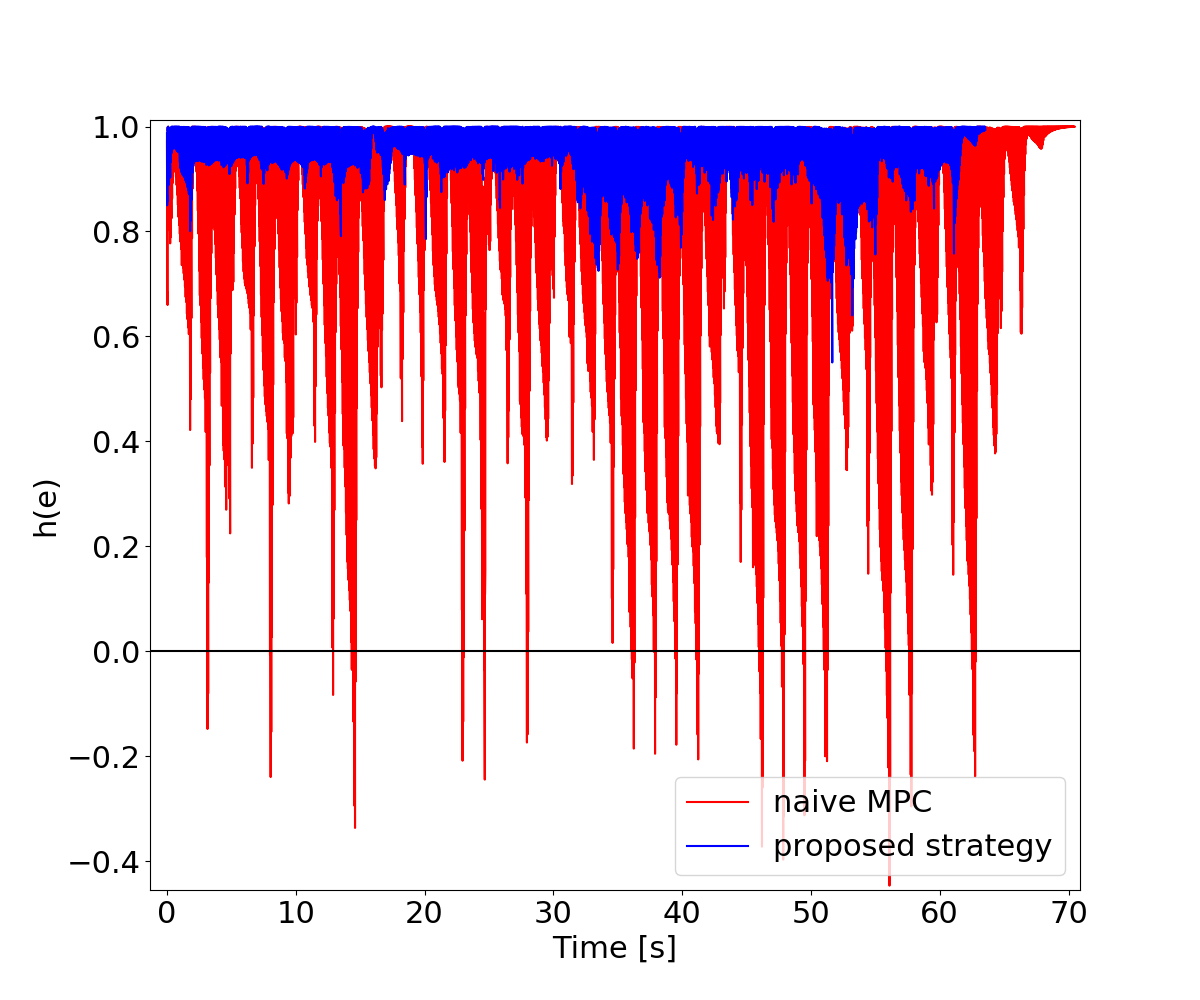}
    \caption{Comparison between the barrier function associated with the proposed strategy and a naive MPC that is based on the linearized dynamics and it is synthesized without robustifying the constraints. As shown in the figure, when the low-level controller is not used the difference between the planner trajectory and the MPC trajectory grows and, as a results, the barrier function~\eqref{eq:barrier} becomes negative. }
    \label{fig:h}
\end{figure}

\subsection{Simulation}

We implemented the proposed strategy in our high-fidelity Robotic Operating System (ROS) simulator. Figure~\ref{fig:clSim} shows the locations of the uncertain and goal regions. The code can be found at \texttt{\url{https://github.com/DrewSingletary/segway_sim}}, please check the \texttt{REAME.md} to replicate our results. In this example the goal regions $\mathcal{G}_1$ and $\mathcal{G}_2$ may contain a science sample with probability $0.6$ and $0.4$, respectively. Whereas, regions $\mathcal{R}_1$ and $\mathcal{R}_2$ may be traversable with probability $0.5$ and $0.1$, as shown in Figure~\ref{fig:prob}.

Figure~\ref{fig:clSim} shows the closed-loop trajectory of the Segway. We notice that the controller explores the uncertain region $\mathcal{R}_1$, which in this example is traversable and afterwards it reaches the goal regions $\mathcal{G}_1$. As shown in Figure~\ref{fig:prob}, at the high-level time $k=19$ the controller figures out that the goal cell $\mathcal{G}_1$ does not contain a science sample and, consequently, the probability of mission success drops. Afterwards, the controller steers the Segway to the traversable region $\mathcal{R}_2$ and to the goal regions $\mathcal{G}_2$. In this example, the goal regions $\mathcal{G}_2$ contains a science sample and the mission is completed successfully, as shown in Figure~\ref{fig:prob}. 

The mid-level is discretized for $T = 50$ ms and the low-level at $1$ kHz. Figure~\ref{fig:comp} shows the computational time associated with mid-level and low-level control actions. It takes on average $12$ ms and at most about $30$ ms to compute the mid-level control action $u_m(t)$--thus the mid-level planner runs in real-time. Furthermore, we notice that it takes less than $1$ ms to compute the low-level action $u_l(t)$. 

Finally, we analyze the evolution of the barrier function~\eqref{eq:barrier}, which measures the difference between the trajectory $x(t)$ of system~\eqref{eq:sysModel} and the reference trajectory $\bar x(t)$ associated with nominal model~\eqref{eq:referenceModel}, which is computed by the mid-level planner. We compared the proposed strategy with a naive MPC which is synthesized as in~\eqref{eq:ftocp}, but without taking into account the effect of the tracking error, i.e., we do not tighten the constraints and we set $x_{i|i}=x(t)$. Figure~\ref{fig:h} shows the evolution of the barrier function for the proposed strategy and the naive MPC. We notice that when the low-level controller is not used, the barrier function becomes negative and in general has a lower magnitude. Therefore, this figure shows the advantage of the proposed hierarchical control architecture, where the low-level high-frequency controller is leveraged to track the reference trajectory. Indeed, this high-frequency feedback is used to modify the mid-level control actions, as shown in Figure~\ref{fig:inputs}. The mid-level control action is updated at $20$ Hz and the low-level input at $1$ kHz. Notice that after the update of the mid-level input, the contribution of the low-level input towards the total control action is limited. However, as time progresses the accuracy of the linearization used to plan the reference trajectory decreases and the effort required by low-level controller to track the trajectory computed by the mid-level planner increases.

\begin{figure}[t!]
    \centering
	\includegraphics[trim= 3mm 3mm 3mm 3mm, clip, width= 1.0\columnwidth]{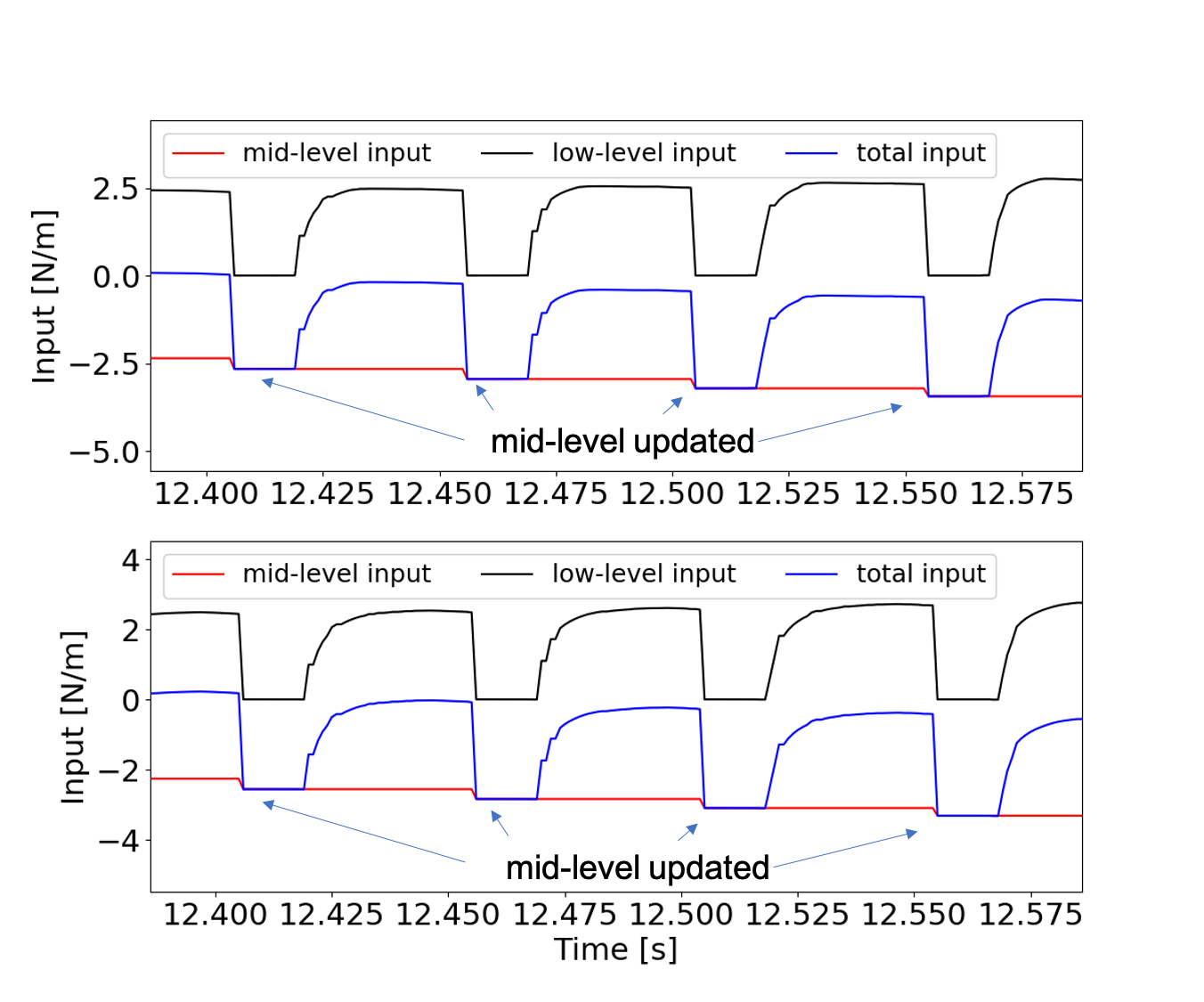}
    \caption{Input torque sent to the right (top) and left (bottom) motor over a period of $0.2$ second. The mid-level input is updated at $20$ Hz, whereas the low-level action is updated at $1$ kHz.}
    \label{fig:inputs}
\end{figure}

\begin{figure}[t!]
    \centering
	\includegraphics[trim= 3mm 3mm 3mm 3mm, clip, width= 1.0\columnwidth]{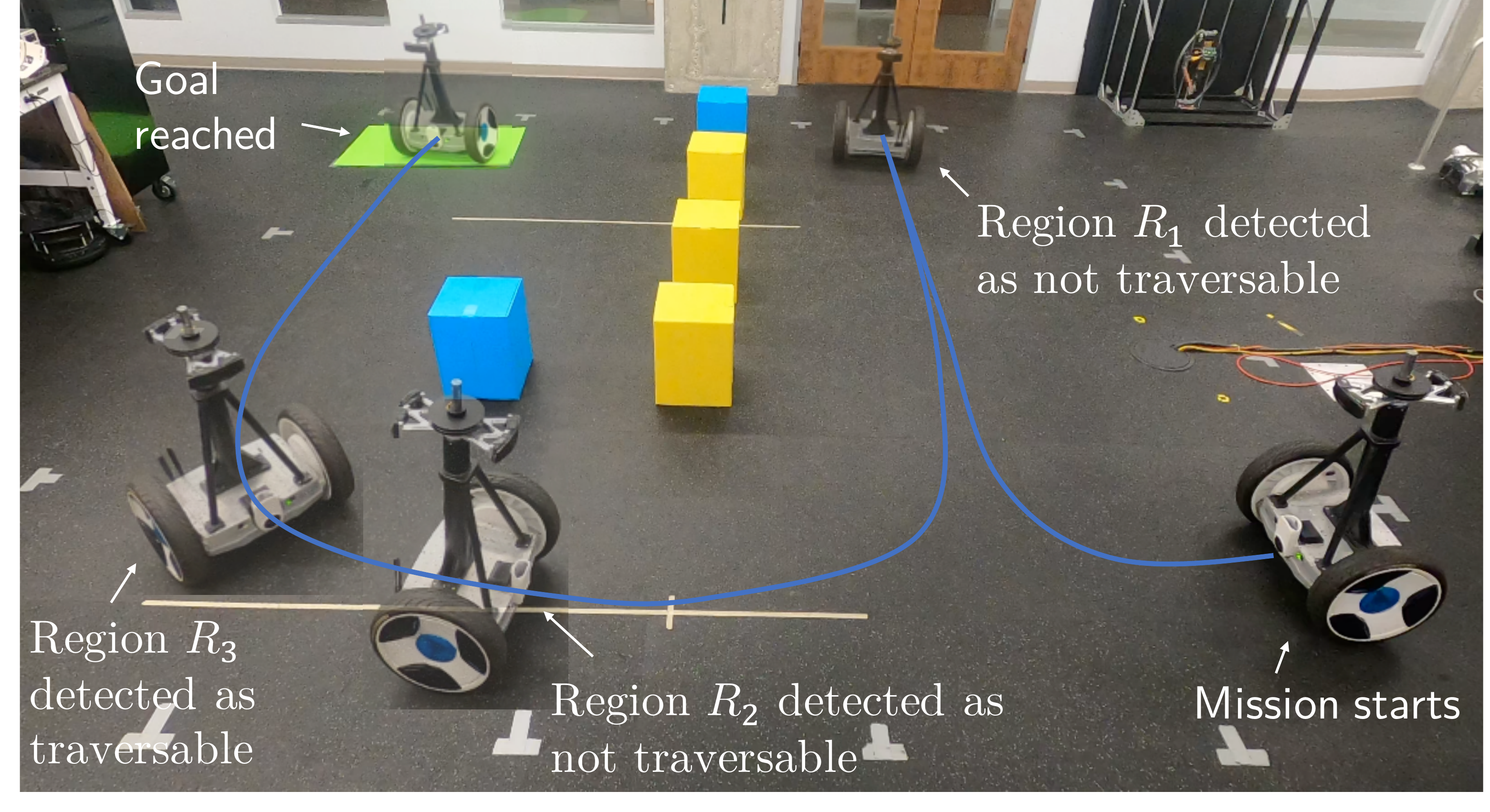}
    \caption{Closed-loop trajectory during the experiment. The Segway first explores the uncertain regions $\mathcal{R}_1$, $\mathcal{R}_2$, and $\mathcal{R}_3$ and afterwards it reaches the goal region.}
    \label{fig:clExp}
\end{figure}

\begin{figure}[t!]
    \centering
	\includegraphics[trim= 3mm 3mm 3mm 3mm, clip, width= 1.0\columnwidth]{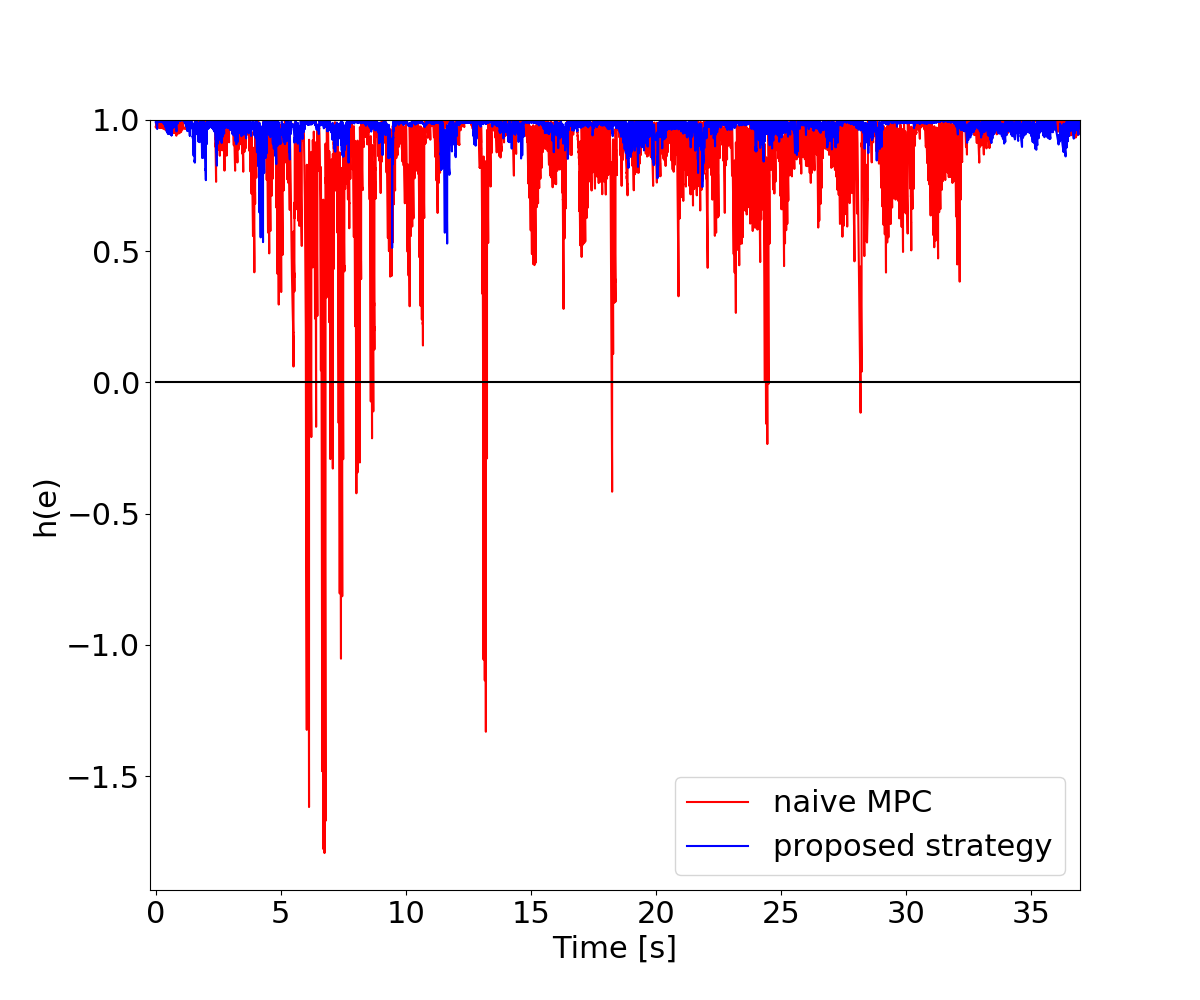}
    \caption{Experimental comparison between the barrier function associated with the proposed strategy and a naive MPC which is based on the linearized dynamics. Also in this case, when the low-level controller is not used, the difference between the planner trajectory and the MPC trajectory grows and, as a results, the barrier function defined in~\eqref{eq:barrier} becomes negative. }
    \label{fig:hExp}
\end{figure}

\subsection{Experiment}
We implemented the proposed multi-rate hierarchical control strategy on the Segway-like robot shown in Figure~\ref{fig:envScheme}.
State estimation is based on wheel encoders and IMU data from a VectorNav VN-100. The state estimate and the low-level control action $u_l$ are computed at $800$ Hz on the Segway, which is equipped with an ARM Cortex-A57 (quad-core) @ 2 GHz CPU running the ERIKA3 RTOS. On the other hand, the mid-level planner discretized at $20$ Hz and the high-level decision maker run on a desktop with an Intel Core i7-8700 CPU (6-cores) @ 3.7 GHz CPU, which sends the reference trajectory $\bar x$ and the reference input $u_m$ via WiFi.

Figure~\ref{fig:envScheme} shows the location of the three uncertain regions $\mathcal{R}_1$, $\mathcal{R}_2$ and $\mathcal{R}_3$ which may be traversable with probability $0.9$, $0.3$ and $0.2$, respectively. In this example, we assume that the goal region $\mathcal{G}_1$ contains the science sample with probability 1. Figure~\ref{fig:clExp} shows the closed-loop trajectory. First, the controller explores region $\mathcal{R}_1$ which is not traversable, and afterwards it steers the Segway towards regions $\mathcal{R}_2$ and $\mathcal{R}_3$. After collecting observations about the environment, the controller detects that region $\mathcal{R}_2$ is not traversable and that region $\mathcal{R}_3$ is free space that the Segway can navigate through to reach the goal region $\mathcal{G}_1$. A video of the experiment and comparison with a naive MPC can be found at \texttt{\url{https://www.youtube.com/watch?v=Q-Mm0ywPh_I}}.

Figure~\ref{fig:hExp} shows the evolution of the control barrier function~\eqref{eq:barrier}. We compare the proposed strategy with a naive MPC which is designed as in~\eqref{eq:ftocp}, but without robustifying the constraint sets and setting $x_{i|i}= x(t)$. Also in this case, when the high-frequency low-level controller is not active, the barrier function becomes negative meaning that the error $e$ does not belong to the safe set $\mathcal{E}$, i.e., $e(t) \notin \mathcal{E}$ for all $t \in \Rp$. This result highlights the importance of the low-level high-frequency feedback from the CLF-CBF QP, which compensates for the model mismatch at the planning layer. Indeed, the MPC planner uses a linearized and discretized model, which is a first order approximation of the true dynamics. This approximation is accurate only at the discrete time instances when the MPC input is computed. To compensate for this model inaccuracy, the low-level CLF-CBF QP tracking controller computes the high-frequency component $u_l(t)$ that is added to the mid-level piecewise constant input $u_m(t)$, as shown in Figure~\ref{fig:inputsExp}.

\begin{figure}[t!]
    \centering
	\includegraphics[trim= 3mm 3mm 3mm 3mm, clip, width= 1.0\columnwidth]{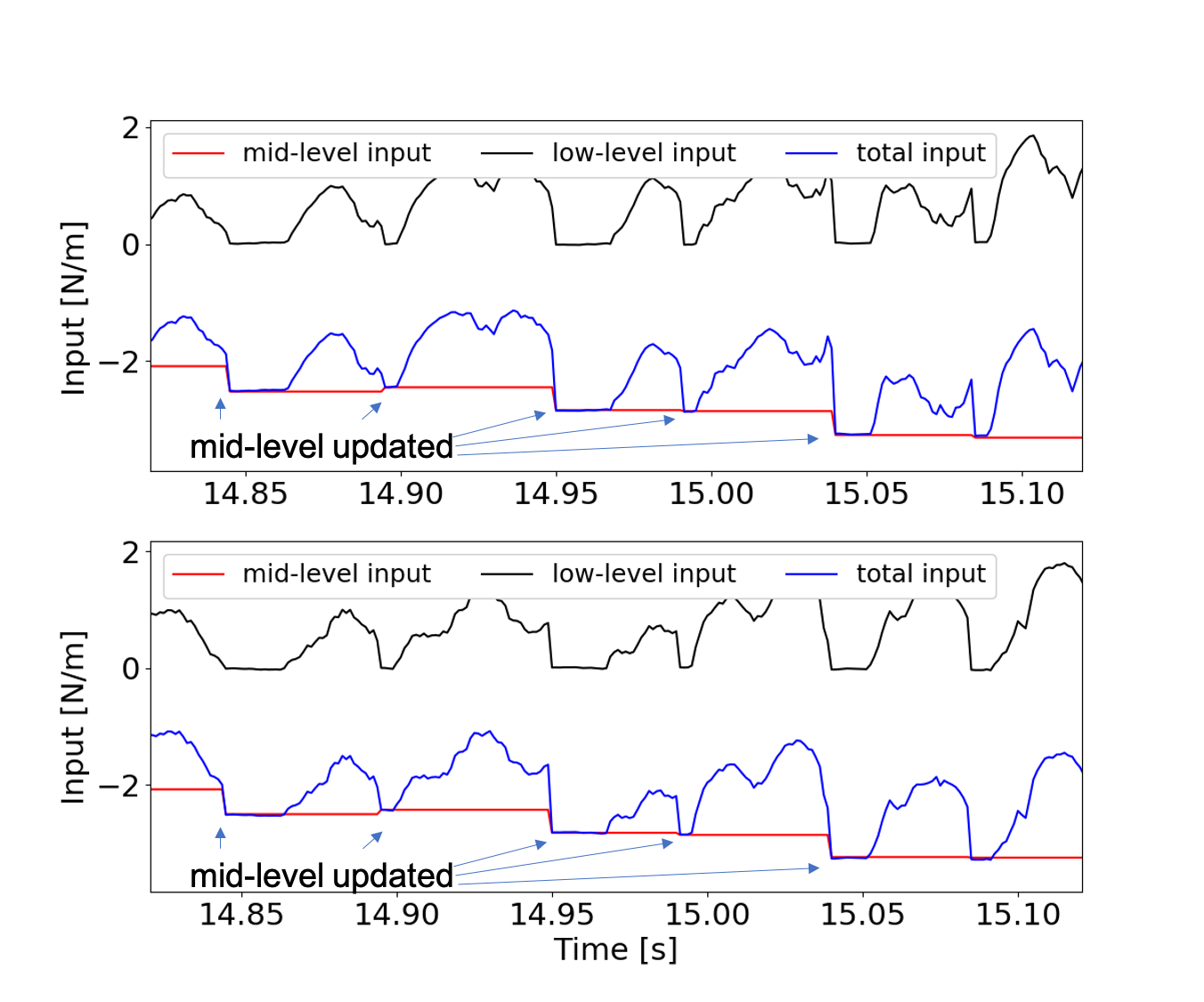}
    \caption{Experimental results. Input torque sent to the right (top) and left (bottom) motor over a period of $0.3$ seconds. The mid-level input is updated at $20$ Hz, whereas the low-level action is updated at $800$ Hz. Notice that the total input is the summation of the low-level and mid-level control inputs.}
    \label{fig:inputsExp}
\end{figure}

\balance 
\section{Conclusions}
In this paper we presented a multi-rate hierarchical control architecture for navigation tasks in partially observable environments. At the lowest level we leverage a CLF-CBF QP, which is used to track a reference trajectory within some error bounds. The reference trajectory is computed by a mid-level planner which leverages an MPC with time-varying terminal components. The feasibility of the MPC planner is guaranteed via a contingency scheme and a local reachability assumption on the planning model. Finally, at the highest level of abstraction, we showed how to model the system-environment interaction using a MOMDP and we proposed an algorithm to update the MPC time-varying components. 
The effectiveness of the proposed strategy is shown on navigation examples, where a Segway-like robot has to find science samples, while avoiding partially observable obstacles.

\section{Acknowledgements}
The authors would like to thank Geoffroy le Courtois du Manoir for helping with experiments and anonymous reviewers for constructive suggestions. 

\bibliographystyle{IEEEtran}
\bibliography{mybib}

\end{document}